\documentclass[a4paper,twocolumn,11pt,unpublished]{quantumarticle}
\pdfoutput=1

\usepackage[utf8]{inputenc}
\usepackage[english]{babel}
\usepackage[T1]{fontenc}
\usepackage{amsmath,amssymb,amsfonts,amsthm,mathtools}
\usepackage{graphicx}
\usepackage{booktabs}
\usepackage{array}
\usepackage{enumitem}
\usepackage{multirow}
\usepackage{float}
\usepackage{placeins}
\usepackage{xcolor}
\usepackage{hyperref}
\usepackage[nameinlink,capitalise,noabbrev]{cleveref}
\usepackage[numbers,sort&compress]{natbib}

\hypersetup{colorlinks=true, linkcolor=blue, citecolor=blue, urlcolor=blue, hypertexnames=false}
\newtheorem{theorem}{Theorem}[section]
\newtheorem{lemma}[theorem]{Lemma}

\newtheorem{definition}[theorem]{Definition}
\newtheorem{assumption}[theorem]{Assumption}
\newtheorem{remark}[theorem]{Remark}
\crefname{assumption}{Assumption}{Assumptions}
\Crefname{assumption}{Assumption}{Assumptions}
\crefname{theorem}{Theorem}{Theorems}
\Crefname{theorem}{Theorem}{Theorems}
\crefname{lemma}{Lemma}{Lemmas}
\Crefname{lemma}{Lemma}{Lemmas}
\crefname{remark}{Remark}{Remarks}
\Crefname{remark}{Remark}{Remarks}

\crefalias{lemma}{lemma}

\newcommand{\ket}[1]{\left|#1\right\rangle}
\newcommand{\bra}[1]{\left\langle #1\right|}

\newcommand{\norm}[1]{\left\lVert #1\right\rVert}

\newcommand{\Oh}{\mathcal{O}}
\newcommand{\Oht}{\widetilde{\mathcal{O}}}

\newcommand{\Ndof}{\mathcal N}

\newif\ifshowrevisions
\showrevisionstrue

\newcommand{\maybeincludegraphics}[2][]{%
  \IfFileExists{#2}{\includegraphics[#1]{#2}}{%
    \fbox{\begin{minipage}[c][0.18\textheight][c]{0.85\linewidth}
    \centering
    \textbf{Missing figure file}\par
    {\ttfamily\detokenize{#2}}\par\medskip
    Replace this placeholder by the original figure when compiling the final version.
    \end{minipage}}%
  }%
}

\title{Explicit block-encodings for biharmonic boundary-value problems}

\author[1,2,3]{Chuwen Ma}
\email{cwma@math.ecnu.edu.cn}
\author[4]{Zihao Tang}
\email{momo\_0@sjtu.edu.cn}

\affil[1]{School of Mathematical Sciences,  East China Normal University, Shanghai 200241, China}
\affil[2]{Key Laboratory of MEA, Ministry of Education, East China Normal University, Shanghai 200241, China}
\affil[3]{Shanghai Key Laboratory of PMMP, East China Normal University, Shanghai 200241, China}
\affil[4]{School of Mathematical Sciences, Shanghai Jiao Tong University, Shanghai 200240, China}
\date{June 8, 2026}

\begin{document}
\maketitle

\begin{abstract}
The biharmonic equation is a prototypical fourth-order partial differential equation whose high-dimensional discretization suffers from rapidly growing degrees of freedom and severe ill-conditioning. We develop QSVT--VTAA quantum linear-system algorithms by constructing explicit block-encodings tailored to periodic, simply supported, and Dirichlet--Neumann boundary conditions. For periodic and simply supported problems, Fourier and sine-transform diagonalizations yield augmented Poisson systems with the condition-number scaling of a second-order operator. For Dirichlet--Neumann problems, we introduce a second-order boundary-corrected finite-difference discretization, establish mesh-independent stability, and construct an explicit block-encoding of the resulting nonsymmetric matrix. We also formulate a coupled-Laplace system with additional boundary unknowns and characterize its complexity in terms of the condition number of the complete augmented matrix. 
The analysis covers discretization error, block-encoding normalization, gate complexity, and solution extraction under an amplitude-input and quantum-state-output model. Numerical experiments validate the proposed discretizations and the corresponding linear solves.
\end{abstract}

\bigskip


\section{Introduction}

Partial differential equations (PDEs) are fundamental mathematical models in
scientific and engineering computing. Among them, the biharmonic equation
\begin{equation}\label{eq:biharmonic-general}
	\Delta^2 u = f
\end{equation}
is a canonical fourth-order elliptic equation. It appears in a wide range of
applications, including continuum mechanics, fluid dynamics, and materials
science \cite{meleshko2003selected,martinodbiharmonic}. In the Kirchhoff--Love
theory of thin plates, the unknown \(u\) represents the transverse displacement
of an elastic plate, and the biharmonic operator describes bending effects,
moments, and stress distributions \cite{ziegler2012mechanics,eschenauer2012applied}.
In low-Reynolds-number fluid dynamics, the stream-function formulation of
creeping flows also leads to biharmonic equations, which are used in the study
of Stokes flows, lubrication, and microfluidic models
\cite{anderson1995computational}. It also appears in CAD \cite{Crane2013} and image processing \cite{Press2007}. Compared with second-order elliptic equations
such as the Poisson equation, the biharmonic equation contains higher-order
spatial derivatives and requires two boundary conditions on each boundary
component. This makes both its mathematical analysis and numerical
discretization substantially more delicate.

Classical approaches include finite-difference, finite-element, spectral, and
mixed formulations. Conforming finite elements generally require
\(C^1\)-continuous basis functions
\cite{Argyris_Fried_Scharpf_1968,bellelement,Bogner1965TheGO},
whereas mixed formulations reduce the regularity requirement by introducing
additional variables and enlarged systems \cite{CIARLET1974125}.
For finite-difference methods, the treatment of boundary conditions often
requires ghost values or special boundary closures. These modifications may
destroy the tensor-product structure of the interior operator and can produce
nonsymmetric matrices
\cite{fdbiharminic,bjorstad1983fast,STEPHENSON198465}.
The computational cost becomes particularly significant in high dimensions:
with \(N\) grid points in each coordinate direction, a tensor-product
discretization contains \(N^d\) unknowns.


Quantum linear-system algorithms provide a framework for preparing quantum
states proportional to the solutions of large discretized systems. Starting
from the HHL algorithm~\cite{Harrow_2009}, subsequent developments based on
linear combinations of unitaries and quantum singular value transformation
(QSVT) have substantially improved the dependence on the target accuracy and
the condition number~\cite{Childs2015Quantum,Gilyen2019}.
Schr\"odingerization provides a complementary dynamical approach to linear
algebraic problems and has been applied to quantum implementations of
classical iterative methods and to multilevel preconditioning for PDE
discretizations~\cite{JinLiu2024,JinLiuMaYu2026}.

In this work, we employ a QSVT linear-system solver combined with
variable-time amplitude amplification (VTAA), whose dependence on the
condition number is nearly linear up to logarithmic factors. For a
discretized PDE, however, the resulting complexity depends not only on the
condition number, but also on the normalization factor and implementation
cost of the block-encoding. It is therefore essential to exploit the
structure of the discrete operator and construct block-encodings with
controlled normalization and gate cost.

The present work focuses on linear-system-based quantum algorithms for the
biharmonic equation. After discretization, the stationary biharmonic problem
naturally becomes a linear system. A general QSVT linear-system solver can in
principle be applied once a suitable block-encoding of the system matrix is
available. However, for specific  PDEs, the construction of such a
block-encoding is often the main nontrivial step. This issue is particularly
important for the biharmonic equation. The fourth-order operator, the presence
of different boundary conditions, and the boundary correction terms appearing
in finite difference discretizations lead to structured matrices that are not
always simple sparse Laplacians. In some cases, the matrices may be
nonsymmetric after the incorporation of the numerical boundary conditions. Therefore, it is not sufficient to
only invoke an abstract QSVT solver; one must construct the required
block-encodings and quantum circuits explicitly.


For the biharmonic equation, the form of this block-encoding depends strongly
on the boundary conditions. Under periodic or simply supported boundary
conditions, introducing
$ v=-\Delta u $
leads to an augmented system involving two second-order operators. Fourier or
sine-transform diagonalization can then be used to avoid directly inverting
the squared Laplacian. Dirichlet--Neumann conditions require a different
treatment. The boundary values of \(v\) are not prescribed, and a consistent
boundary closure for the fourth-order operator changes the interior
tensor-product matrix and may make the resulting discretization
nonsymmetric. Alternatively, the unknown boundary values of \(v\) can be
included as additional variables in a coupled-Laplace formulation.

\noindent\textbf{Relation to prior work.}
The present work builds on QSVT-based quantum linear-system algorithms
\cite{Childs2015Quantum,Gilyen2019} and their combination with variable-time
amplitude amplification, as well as on recent explicit block-encodings for
elliptic boundary-value problems \cite{kharazi2025explicit}. Our focus is the
structured linear algebra produced by biharmonic discretizations. In
particular, we construct boundary-condition-dependent block-encodings for
augmented Poisson operators, nonsymmetric boundary-corrected fourth-order
matrices, and coupled interior--boundary systems. The normalization factors and
implementation costs of these block-encodings are then incorporated into the
QSVT--VTAA gate-complexity analysis.

\paragraph{Main contributions.}
The main results are as follows.
\begin{itemize}[leftmargin=1.7em]

\item For periodic boundary conditions, we use a Fourier spectral
discretization on the mean-zero subspace and represent the biharmonic problem
by an augmented Poisson matrix \(P_N\). We construct an explicit diagonal
block-encoding of \(P_N\). The augmented formulation has condition number
\(O(dN^2)\), compared with \(O(d^2N^4)\) for direct inversion of the squared
Fourier Laplacian.

\item For simply supported boundary conditions, we diagonalize the
finite-difference Laplacian by the quantum discrete sine transform and
construct an explicit block-encoding of the corresponding augmented matrix
\(D_N\). Its condition number scales as \(O(N^2)\).

\item For Dirichlet--Neumann boundary conditions, including the clamped case,
we introduce a second-order one-sided boundary closure for the fourth-order
operator. We establish a mesh-independent stability estimate through the
symmetric part of the resulting nonsymmetric matrix and construct an exact,
norm-matched block-encoding of the boundary correction. The condition number
of the complete corrected matrix is \(O(N^4)\).

\item We also formulate a coupled-Laplace system in which the unknown boundary
values of \(v=-\Delta u\) are stored in a separate boundary register. The
interior--boundary operators are block-encoded using partial-isometry
constructions and assembled into a block-encoding of the complete augmented
matrix \(K\). The corresponding quantum complexity is expressed in terms of
\(\kappa(K)\); a two-dimensional numerical experiment indicates the empirical
scaling \(\kappa(K)\approx N^3\).

\end{itemize}

For each formulation, we analyze the block-encoding normalization and
implementation cost, the resulting QSVT--VTAA gate complexity, and the cost of
extracting the physical solution component under an amplitude-input and
quantum-state-output model. Discretization errors are treated separately from
the quantum state-preparation errors. Small-scale numerical experiments verify
the proposed discretizations and the corresponding block-encoded linear
systems. The principal finite-dimensional matrices and fixed-grid complexity
bounds are summarized in \Cref{tab:main-results}.


\begin{table*}[t]
\begingroup
\scriptsize
\centering
\resizebox{\textwidth}{!}{%
\begin{tabular}{@{}lllll@{}}
\toprule
Boundary condition & Discretization & Block-encoded matrix & Classical complexity & Fixed-grid gate complexity \\
\midrule
Periodic & Fourier spectral & \(P_N\), Eq.~\eqref{eq:expanded-system} & \(\Oht(N^d)\) & \(\Oht(d^2N^2\sqrt{1+\rho_N^2})\) \\
Simply supported & FD + DST & \(D_N\), Eq.~\eqref{eq:dst-system} & \(\Oht(N^d)\) & \(\Oht(dN^2\sqrt{1+\rho_N^2})\) \\
Dirichlet--Neumann & Corrected FD & \(\widetilde A\), Eq.~\eqref{eq:dn-system} & \(\Oht(N^{d + 4})\) & \(\Oht(dN^4)\) \\
Dirichlet--Neumann & Coupled Laplace & \(K\), Eq.~\eqref{eq:split-linear} & \(\Oht(N^{d + 3})\) &  \(\Oht(dN^3\sqrt{1+\rho_N^2})\) \\
\bottomrule
\end{tabular}}
\caption{Roadmap of the finite-dimensional matrices used in the QSVT-VTAA solver.  Here \(N\) is the one-dimensional resolution and \(\Ndof=N^d\) is the total number of scalar grid unknowns.  For the augmented Poisson formulations, \(\rho_N\) denotes the ratio between the auxiliary component and the solution component; its behavior in smooth and worst-case regimes is discussed after the corresponding theorems. For the linear systems derived from the last two difference schemes, we adopt the complexity of the GMRES method because the corresponding differential matrices are non-symmetric. Notably, the condition number of the difference matrix $K$ from the coupled Laplace system is estimated via numerical trials, while those of the other difference matrices enjoy strict theoretical guarantees. }
\label{tab:main-results}
\endgroup
\end{table*}

\paragraph{Organization.}
\Cref{sec:prelim} reviews block-encodings, the LCU construction, the
QSVT--VTAA linear-system solver, and block-encodings of discrete Laplace
operators. \Cref{sec:periodic} treats periodic boundary conditions using a
Fourier spectral discretization, and \Cref{sec:ss} considers simply supported
boundary conditions using finite differences and the quantum discrete sine
transform. \Cref{sec:dn} develops the direct boundary-corrected discretization
for Dirichlet--Neumann data, while \Cref{sec:coupled} presents the
coupled-Laplace formulation. Numerical results are reported in
\Cref{sec:numerics}, followed by conclusions in \Cref{sec:conclusion}.

\paragraph{Notation.}
We use \(\Oh(\cdot)\) for asymptotic upper bounds and \(\Oht(\cdot)\) when logarithmic factors are suppressed.  The spatial dimension is denoted by \(d\).  Throughout the PDE sections, \(N\) denotes the one-dimensional resolution: in periodic spectral discretizations it is the number of collocation points per coordinate, while in finite-difference sections it is the number of interior grid points per coordinate unless stated otherwise.  The total number of scalar grid unknowns is denoted by \(\Ndof=N^d\). 
The mesh size is \(h\).
Continuous functions are written in italic, for example \(u\), \(v\), and \(f\).  Discrete grid vectors are written in bold, for example \(\mathbf u\), \(\mathbf v\), and \(\mathbf f\); their Fourier or sine coefficients are denoted by hats, for example \(\widehat{\mathbf u}\).  
For a matrix \(A\), \(\norm{A}\) denotes its spectral norm, and \(\kappa(A)\) its condition number on the relevant nonsingular subspace.  The normalized quantum state associated with a nonzero discrete vector \(\mathbf b\) is denoted by \(\ket{\mathbf b}=\mathbf b/\norm{\mathbf b}\). In the end-to-end estimates, the discretization and quantum-state errors are each
chosen to be of order $\varepsilon$, so that the total error remains $O(\varepsilon)$.

\begin{assumption}[Input model and output model]\label{ass:input-output}
The normalized quantum states associated with the discrete
right-hand side and boundary data are assumed to be efficiently
preparable. Their preparation costs are not included in the
complexity estimates below. This assumption is intended for data
with a succinct and efficiently evaluable classical representation,
as is the case for many smooth structured inputs. The algorithms
output normalized solution states or normalized augmented states.
\end{assumption}

\begin{assumption}[Null spaces]\label{ass:nullspace}
When the continuous or discrete operator has a null space, we restrict the problem to the orthogonal complement of the null space or solve the corresponding pseudoinverse problem.  For periodic boundary conditions, this means that the right-hand side has zero mean and the solution is chosen with zero mean.
\end{assumption}


\section{Preliminaries}\label{sec:prelim}

This section recalls the quantum linear-algebra tools used in this paper. We
only review the ingredients needed for the later construction: block-encodings,
the linear combination of unitaries (LCU) framework, quantum singular value
transformation (QSVT) for linear systems, QSVT method combined with variable time amplitude amplification (VTAA) and an explicit block-encoding of the
finite-difference Laplacian. For more complete introductions to these tools, we
refer to \cite{Gilyen2019,Low_2017,Chakraborty_2023,lin2022lecturenotesquantumalgorithms}.

\subsection{Block-encodings and LCU}
\label{subsec:lcu}
Block-encoding is a standard way to represent a non-unitary matrix as a
sub-block of a larger unitary matrix. This framework is widely used in quantum
linear algebra, Hamiltonian simulation, and QSVT-based algorithms
.

\begin{definition}[Block-encoding]\label{def:block-encoding}
	Let \(A\in\mathbb C^{M\times M}\), \(M=2^n\), and let
	\(\alpha>0\), \(a\in\mathbb N\), \(\varepsilon\ge 0\). An \((n+a)\)-qubit
	unitary \(U_A\) is called an \((\alpha,a,\varepsilon)\)-block-encoding of
	\(A\) if
	\begin{equation}\label{eq:block-encoding}
		\left\|
		A-\alpha(\bra{0^a}\otimes I_M)U_A(\ket{0^a}\otimes I_M)
		\right\|
		\le \varepsilon .
	\end{equation}
	When \(\varepsilon=0\), the block-encoding is called exact.
\end{definition}

The normalization factor \(\alpha\) plays an important role in QSVT-based
linear-system algorithms. A smaller ratio \(\alpha/\|A\|\) generally leads to
a more efficient algorithm.

We next recall the LCU construction \cite{Gilyen2019}. Let
\begin{equation}
	T=\sum_{j=0}^{K-1}\alpha_j U_j,\qquad \alpha_j>0,
\end{equation}
where each \(U_j\) is unitary. Any complex phase in a coefficient can be
absorbed into the corresponding unitary \(U_j\). Let
\[
\alpha_\Sigma=\sum_{j=0}^{K-1}\alpha_j .
\]
Assume \(K\le 2^a\), after padding by zero coefficients if necessary. Define
the selection oracle
\begin{equation}
	\mathrm{SELECT}(U)
	=
	\sum_{j=0}^{K-1}\ket{j}\bra{j}\otimes U_j,
\end{equation}
and the state-preparation oracle
\begin{equation}\label{eq:prep LCU}
	\mathrm{PREP}\ket{0^a}
	=
	\frac{1}{\sqrt{\alpha_\Sigma}}
	\sum_{j=0}^{K-1}\sqrt{\alpha_j}\ket{j}.
\end{equation}
Then
\begin{equation}\label{eq:lcu-unitary}
	U_T
	=
	(\mathrm{PREP}^{\dagger}\otimes I)\,
	\mathrm{SELECT}(U)\,
	(\mathrm{PREP}\otimes I)
\end{equation}
is an \((\alpha_\Sigma,a,0)\)-block-encoding of \(T\), since
\begin{equation}
	(\bra{0^a}\otimes I)U_T(\ket{0^a}\otimes I)
	=
	\frac{T}{\alpha_\Sigma}.
\end{equation}

The corresponding circuit layout is shown in \cref{fig:lcu}.

In the LCU constructions, the coefficient-state preparation oracle
\(\mathrm{PREP}\) is implemented using standard state-preparation routines.
For a general coefficient vector \((\alpha_0,\ldots,\alpha_{K-1})\), this
amounts to preparing \eqref{eq:prep LCU}.
A generic \(K\)-dimensional coefficient state can be prepared by the
tree-based amplitude-encoding procedure \cite{m_tt_nen_2004}. In the concrete
LCU constructions used in this paper, these coefficient registers are either
constant-size, logarithmic-size, or uniform superpositions over spatial
directions, and their preparation cost is included in the stated
block-encoding costs.

\begin{figure}[h]
	\centering
	\maybeincludegraphics[width=0.78\linewidth]{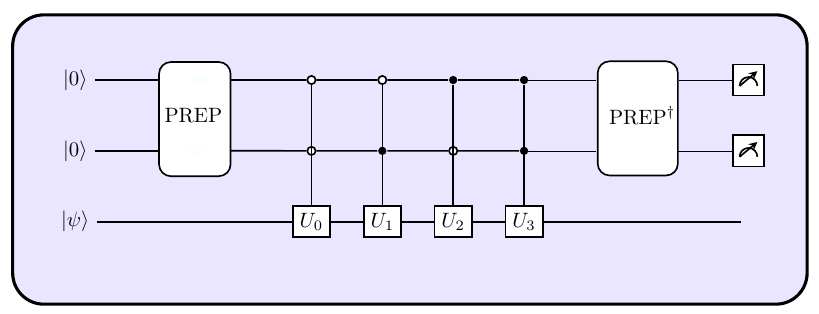}
	\caption{LCU construction with \(\mathrm{PREP}\), \(\mathrm{SELECT}\), and
		\(\mathrm{PREP}^{\dagger}\).}
	\label{fig:lcu}
\end{figure}

We will repeatedly use the following standard arithmetic rules for
block-encodings \cite{Gilyen2019,lin2022lecturenotesquantumalgorithms}.

\begin{lemma}[Arithmetic of block-encodings]\label{lem:block-arithmetic}
	Suppose \(U_A\) is an \((\alpha,a,\varepsilon_A)\)-block-encoding of \(A\)
	with cost \(T_A\), and \(U_B\) is a \((\beta,b,\varepsilon_B)\)-block-encoding
	of \(B\) with cost \(T_B\). Then:
	\begin{enumerate}[leftmargin=2em]
		\item \(A+B\) admits a block-encoding with normalization factor
		\(\alpha+\beta\), using LCU. The gate cost is
		\(\mathcal O(T_A+T_B)\), up to the cost of the LCU state preparation and
		the necessary ancilla padding.
		
		\item \(AB\) admits a block-encoding with normalization factor
		\(\alpha\beta\), obtained by composing the two block-encodings. The gate
		cost is \(\mathcal O(T_A+T_B)\). In the approximate case, the error is at
		most
		\[
		\alpha\varepsilon_B+\beta\varepsilon_A+\varepsilon_A\varepsilon_B .
		\]
	\end{enumerate}
\end{lemma}

\subsection{QSVT-VTAA linear-system solver}
\label{subsec:qsvt-ls}
\label{subsec:qsvt-vtaa}

Quantum signal processing (QSP) provides the polynomial-transformation
primitive underlying the linear-system solver used in this paper.
Given the signal operator
\[
O(x)=
\begin{pmatrix}
	x & -\sqrt{1-x^2}\\
	\sqrt{1-x^2} & x
\end{pmatrix},
\qquad x\in[-1,1],
\]
and a phase sequence
\(\Phi=(\varphi_0,\ldots,\varphi_L)\), the QSP sequence
\[
U_{\Phi}(x)
=
e^{\mathrm{i}\varphi_0 Z}
\prod_{j=1}^{L}
\left[
O(x)e^{\mathrm{i}\varphi_j Z}
\right]
\]
implements polynomial entries in \(x\). Conversely, every feasible
polynomial bounded by one on \([-1,1]\), with the required parity,
can be implemented by an appropriate phase sequence
\cite{lin2025mathematicalnumericalanalysisquantum}.
Here \(L\) denotes the polynomial degree.

Quantum singular value transformation (QSVT) lifts this scalar
polynomial transformation to the singular values of a block-encoded
matrix. Let
\[
A=W\Sigma V^\dagger
\]
be a singular-value decomposition. QSVT implements polynomial
transformations of the form
\[
p^\diamond(A)=Wp(\Sigma)V^\dagger.
\]
For a linear system \(A\mathbf x=\mathbf b\), one chooses a bounded
polynomial approximation to a suitably scaled reciprocal function
on the nonzero singular-value interval. With the appropriate adjoint
orientation, QSVT implements
\[
p^\diamond(A^\dagger)
=
Vp(\Sigma)W^\dagger
\approx cA^{-1},
\]
where \(c>0\) is a scaling factor ensuring that the polynomial remains
bounded by one. Conditioned on the success ancillas, the resulting
state is therefore proportional to \(A^{-1}\ket{\mathbf b}\).

The basic QSVT circuit is shown in \cref{fig:qsvt}. It alternates
applications of the block-encoding \(U_A\), its inverse, and controlled
phase rotations
\[
CR_{\varphi}
=
\ket{0^a}\bra{0^a}\otimes e^{-\mathrm{i}\varphi Z}
+
\left(I-\ket{0^a}\bra{0^a}\right)\otimes I .
\]
The precise sign convention for these rotations is absorbed into the
phase sequence \(\Phi\).

\begin{figure}[h]
	\centering
	\maybeincludegraphics[width=1\linewidth]{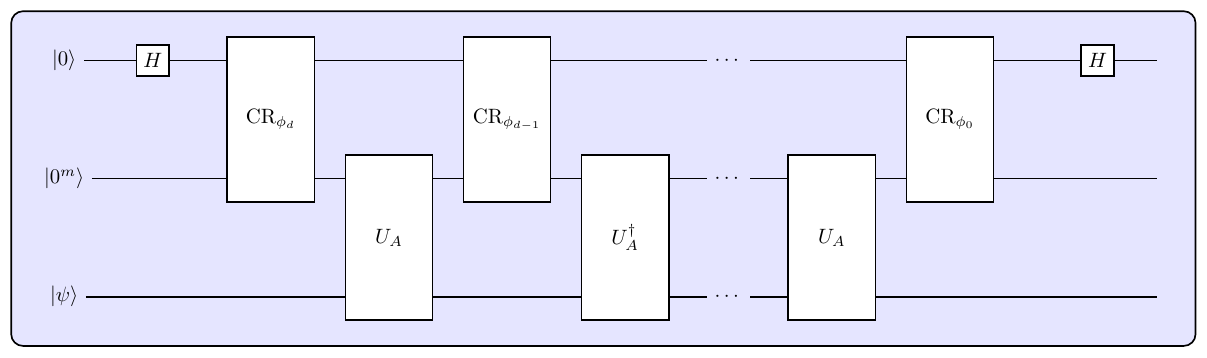}
	\caption{Basic QSVT circuit for applying a polynomial
		transformation to the singular values of a block-encoded matrix.
		When the polynomial approximates a scaled reciprocal function,
		the success component is proportional to
		\(A^{-1}\ket{\psi}\).}
	\label{fig:qsvt}
\end{figure}

A fixed-time QSVT inversion treats all singular-value components at
the worst inversion scale. Although the reciprocal polynomial has
degree nearly linear in the condition number, uniform amplitude
amplification can introduce a second worst-case factor in
\(\kappa\). To avoid this quadratic dependence, all algorithms below
use the QSVT-based variable-time amplitude-amplification method of
\cite{Chakraborty_2023}.

To explain the construction, write
\[
A=\sum_i \sigma_i\ket{w_i}\bra{v_i},
\qquad
\ket{\mathbf b}
=
\sum_i c_i\ket{w_i}.
\]
The desired inverse state is proportional to
\[
A^{-1}\ket{\mathbf b}
=
\sum_i \frac{c_i}{\sigma_i}\ket{v_i}.
\]
The QSVT-VTAA procedure treats the different singular-value scales
through the following variable-stopping-time process.

\begin{enumerate}[leftmargin=2em]
	\item
	\textbf{Dyadic singular-value scales.}
	The singular-value interval is divided into
	\(m=\lceil\log_2\kappa\rceil+1\) dyadic scales with thresholds
	\[
	\gamma_j=2^{-j},
	\qquad j=1,\ldots,m.
	\]
	
	\item
	\textbf{Singular-value discrimination.}
	At stage \(j\), a QSVT-based singular-value discrimination
	procedure tests the unresolved singular components against a
	threshold of order \(\gamma_j\). The result is stored in a clock
	or flag register.
	
	\item
	\textbf{Scale-dependent inversion and stopping.}
	Components assigned to the \(j\)-th scale undergo a localized
	QSVT inversion
	\[
	W(\gamma_j,\varepsilon'),
	\]
	whose query cost is
	\[
	\widetilde O\!\left(\gamma_j^{-1}\right).
	\]
	These branches are then halted, while the unresolved components
	proceed to the next scale. Thus a component with singular value
	\(\sigma_i\) stops at a scale satisfying
	\(\gamma_j=\Theta(\sigma_i)\), up to the transition width of the
	discrimination polynomial.
	
	\item
	\textbf{Variable-time amplitude amplification.}
	VTAA amplifies the success component of the complete
	variable-stopping-time algorithm while accounting for the
	different stopping times of the spectral branches. This avoids
	applying the longest inversion circuit uniformly to every
	component and yields near-linear worst-case dependence on
	\(\kappa\).
\end{enumerate}

The use of QSVT-based singular-value discrimination, rather than
explicit phase estimation, also allows the method to be applied
directly to non-Hermitian block-encoded matrices. This will be
important for the boundary-corrected systems considered below.

\begin{lemma}[QSVT-VTAA linear-system cost
	{\cite{Chakraborty_2023}}]
	\label{thm:qsvt-vtaa}
	Let \(A\) be invertible on a subspace \(\mathcal S\), and suppose
	that the normalized right-hand-side state
	\(\ket{\mathbf b}\) belongs to \(\mathcal S\).
	Assume that \(U_A\) is an exact
	\((\alpha_A,a,0)\)-block-encoding of \(A\), implemented with gate
	cost \(T_A\). Define
	\[
	\beta_A:=\frac{\alpha_A}{\norm{A}},
	\qquad
	\kappa(A):=\norm{A}\norm{A^{-1}},
	\]
	where the inverse and the norms are restricted to
	\(\mathcal S\).
	
	Under Assumption~\ref{ass:input-output}, the QSVT-VTAA
	linear-system algorithm prepares a state that is
	\(\varepsilon\)-close to
	\[
	\frac{A^{-1}\ket{\mathbf b}}
	{\norm{A^{-1}\ket{\mathbf b}}}
	\]
	with gate complexity
	\begin{align*}
	&O\!\left(
	\beta_A\kappa(A)\log\kappa(A)
	\log\frac{\kappa(A)}{\varepsilon}\,T_A
	\right)\\
	=&
	\Oht\!\left(
	\beta_A\kappa(A)T_A\log\frac{1}{\varepsilon}
	\right).
	\end{align*}
	The algorithm uses \(O(\log\kappa(A))\) additional qubits.
\end{lemma}

\begin{remark}
	The general result of \cite{Chakraborty_2023} also allows approximate
	block-encodings, provided that their error is sufficiently small
	relative to the target state accuracy and the condition number.
	The principal block-encodings constructed below are exact at the
	logical-circuit level, so the exact form of
	Lemma \ref{thm:qsvt-vtaa} is sufficient for the present analysis.
	The cost of preparing the normalized right-hand-side state is excluded
	according to Assumption~\ref{ass:input-output}.
\end{remark}

\subsection{Block-encoding of the discrete Laplacian}
\label{subsec:laplacian-be}

We finally recall an explicit block-encoding of the one-dimensional Dirichlet
finite-difference Laplacian, which will be used as a basic subroutine in later
sections. Let
\begin{align}\label{laplace}
L_h^{(1)}
=
\frac{1}{h^2}
\begin{bmatrix}
	2 & -1 & 0 & \cdots & 0 \\
	-1 & 2 & -1 & \ddots & \vdots \\
	0 & -1 & 2 & \ddots & 0 \\
	\vdots & \ddots & \ddots & \ddots & -1 \\
	0 & \cdots & 0 & -1 & 2
\end{bmatrix}_{N\times N}.
\end{align}
Assume \(N=2^n\). Let \(S_1\) be the cyclic increment operator
\begin{align}\label{eq:plus1}
S_1\ket{j}=\ket{j+1\!\!\!\pmod N}.
\end{align}
Define the boundary reflections
\begin{align*}
R_0&=I-2\ket{0}\bra{0},\\
R_{N-1}&=I-2\ket{N-1}\bra{N-1}.
\end{align*}
Then the Dirichlet Laplacian can be written as
\begin{align}\label{eq:lap-shift-form}
L_h^{(1)}
&=\frac{1}{h^2}\Bigl(
2I-\frac12(R_0+I)S_1^{-1}\notag\\
&\hspace{2.2cm}-\frac12(R_{N-1}+I)S_1
\Bigr).
\end{align}
The cyclic increment \(S_1\) can be implemented by reversible addition
circuits, and the reflections \(R_0,R_{N-1}\), illustrated in \cref{fig:reflect}, by multi-controlled phase gates
\cite{Gidney_2018,kharazi2025explicit}. Thus Eq. \eqref{eq:lap-shift-form}
provides an LCU block-encoding of \(L_h^{(1)}\).

\begin{figure}[h]
	\centering
	\maybeincludegraphics[width=1\linewidth]{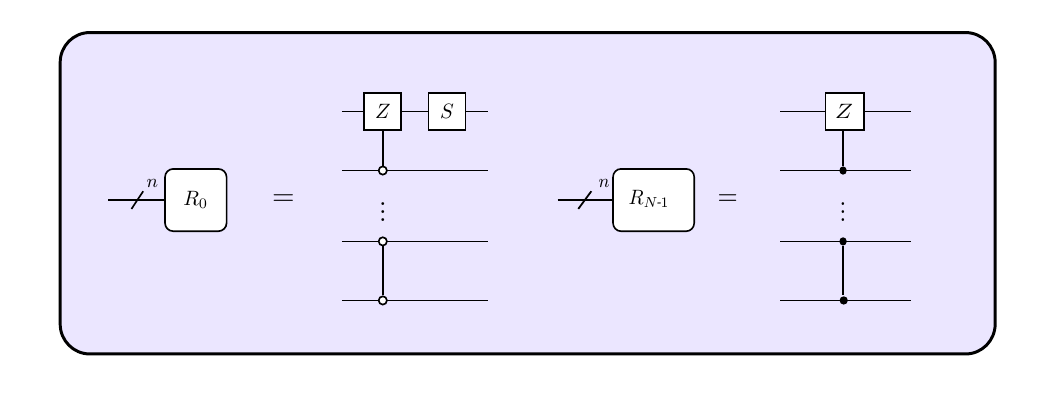}
	\caption{Boundary reflections \(R_0\) and \(R_{N-1}\) used in the
		shift-based block-encoding of the one-dimensional Dirichlet Laplacian.}
	\label{fig:reflect}
\end{figure}

\begin{lemma}[Discrete Laplacian block-encoding]\label{lem:lap-be}
	The one-dimensional Dirichlet Laplacian \(L_h^{(1)}\) admits an exact
	block-encoding with normalization
	\[
	\alpha_1=O(h^{-2})
	\]
	and gate cost \(O(n)\), up to the cost of multi-controlled gates. Moreover, the
	\(d\)-dimensional tensor-product Laplacian
	\[
	L_h^{(d)}
	=
	\sum_{\ell=1}^d
	I^{\otimes(\ell-1)}
	\otimes L_h^{(1)}
	\otimes I^{\otimes(d-\ell)}
	\]
	admits a block-encoding with normalization
	\[
	\alpha_d=O(dh^{-2})
	\]
	and gate cost \(O(dn)\), using an additional LCU register to select the spatial
	direction \(\ell\) \cite{kharazi2025explicit}.
\end{lemma}

\begin{lemma}[Block-encoding of an augmented linear system]
\label{lem:augmented-block-encoding}
Let \(B\) be a matrix acting on \(n\) qubits, and suppose that
\(U_B\) is an exact \((\alpha_B,a,0)\)-block-encoding of \(B\)
with gate cost \(T_B\). Then the augmented matrix
\[
\mathcal P(B)
:=
\begin{bmatrix}
B & -I\\
0 & B
\end{bmatrix}
\]
admits an exact block-encoding with normalization
\[
\alpha_{\mathcal P}=\alpha_B+1
\]
and gate cost
\[
T_{\mathcal P}=T_B+O(1).
\]
\end{lemma}

\begin{proof}
Using
\[
\mathcal P(B)
=
I\otimes B
-\frac12 X\otimes I
-\frac{\mathrm i}{2}Y\otimes I,
\]
the result follows from a three-term LCU construction. The three
LCU normalizations are respectively
\(\alpha_B\), \(1/2\), and \(1/2\), so their sum is
\(\alpha_B+1\). The outer coefficient register has constant
dimension, and its PREPARE, SELECT, and inverse PREPARE circuits
add only \(O(1)\) gates to one use of \(U_B\).
\end{proof}

\section{Periodic boundary conditions}\label{sec:periodic}

We first consider the biharmonic equation on the periodic box
\(\Omega=[-1,1]^d\),
\begin{equation}\label{eq:periodic-biharmonic}
	\begin{cases}
		\Delta^2 u = f, & x\in \Omega,\\
		u \ \text{is periodic in each coordinate}.
	\end{cases}
\end{equation}
Because the periodic Laplace operator has constants in its null space, we
impose the compatibility condition
\[
\int_{\Omega} f(x)\,dx=0
\]
and choose the solution with zero mean. Equivalently, all computations below
are restricted to the mean-zero subspace. Introducing
\[
v=-\Delta u,
\]
the biharmonic equation can be written as two Poisson equations,
\begin{equation}\label{eq:periodic-two-poisson}
	-\Delta v=f,\qquad -\Delta u=v,
\end{equation}
with periodic boundary conditions and zero-mean right-hand sides. Thus the
biharmonic problem can be addressed by one QSVT-VTAA-based linear solver acting on the following expanded system for the positive operator \(-\Delta\):
\[
\begin{bmatrix}
    -\Delta & -I\\
    0& -\Delta
\end{bmatrix}\begin{bmatrix}
    u\\v
\end{bmatrix} = \begin{bmatrix}
    0\\f
\end{bmatrix}.
\]

\subsection{Fourier spectral discretization}

Let \(N=2^m\) be the number of grid points in each spatial direction, and let
\[
x_{\ell}=-1+\frac{2\ell}{N},\qquad \ell=0,1,\ldots,N-1,
\]
be the one-dimensional grid points. We use the centered frequency set
\[
\mathcal K_N=\left\{-\frac N2,-\frac N2+1,\ldots,\frac N2-1\right\}.
\]
For \(k=(k_1,\ldots,k_d)\in \mathcal K_N^d\), define
\[
\phi_k(x)=\exp\!\left(i\pi k\cdot x\right).
\]
Then
\[
-\Delta \phi_k=\lambda_k\phi_k,
\qquad
\lambda_k=\pi^2 |k|^2.
\]
The zero mode \(k=0\) is excluded by the mean-zero condition.

Let \(\mathbf f_N\in\mathbb C^{\Ndof}\) be the vector obtained by sampling the continuous right-hand side \(f\) on the Fourier grid and then projecting out the zero mode.  Applying the \(d\)-dimensional quantum Fourier transform to the normalized state \(\ket{\mathbf f_N}\) gives the Fourier coefficient vector \(\widehat{\mathbf f}_N\).  
Let
\(\Lambda_N\) denote the diagonal matrix with entries
\[
(\Lambda_N)_{kk}=\lambda_k=\pi^2 |k|^2,
\qquad k\in \mathcal K_N^d\setminus\{0\}.
\]
The two Poisson problems in \eqref{eq:periodic-two-poisson} become
\begin{equation}\label{eq:periodic-fourier-systems}
	\Lambda_N \widehat{\mathbf v} = \widehat{\mathbf f}_N,
	\qquad
	\Lambda_N \widehat{\mathbf u} = \widehat{\mathbf v} .
\end{equation}
Hence
\[
    \widehat{\mathbf u}=\Lambda_N^{-2}\widehat{\mathbf f}_N
\]
on the mean-zero subspace. The final physical-space solution state is obtained
by applying the inverse quantum Fourier transform.

In quantum computing, we do not directly solve the linear system
\[
\Lambda_N^{2}\widehat{\mathbf u} = \widehat{\mathbf f}_N.
\]
Since $\kappa(\Lambda_{N}^2) = O(N^4)$ and the performance of the QSVT based algorithm on problems with high condition numbers is rather poor  \cite{NOVIKAU2025113767}, one should instead formulate \eqref{eq:periodic-fourier-systems} as a single
augmented system
\begin{align}\label{eq:expanded-system}
\begin{bmatrix}
    \Lambda_N & -I\\
    0 & \Lambda_N
\end{bmatrix}
\begin{bmatrix}
    \widehat{\mathbf u}\\
    \widehat{\mathbf v}
\end{bmatrix}
=
\begin{bmatrix}
    0\\
    \widehat{\mathbf f}_N
\end{bmatrix}.
\end{align}
This formulation decreases the condition number to $O(dN^2)$.
\begin{remark}
One may use the two successive Poisson solves in \eqref{eq:periodic-fourier-systems}. While this approach saves a few qubits during the block-encoding construction stage, it doubles the circuit depth. Furthermore, executing two successive QSVT-VTAAs generates additional garbage states, leading to the amplitude subnormalization of the target state.
\end{remark}

\subsection{Block-encoding and circuit structure}

The overall circuit structure is shown in \cref{fig:periodic}. The QSVT-VTAA procedure solves the linear system \eqref{eq:expanded-system} for $[\widehat{\mathbf u},\widehat{\mathbf v}]^{\mathsf T}$. 
After extracting the \(\widehat{\mathbf u}\)-component, the inverse
quantum Fourier transform produces the physical-space solution state
\(\ket{\mathbf u}\).

\begin{figure}[h]
	\centering
	\maybeincludegraphics[width=1.0\linewidth]{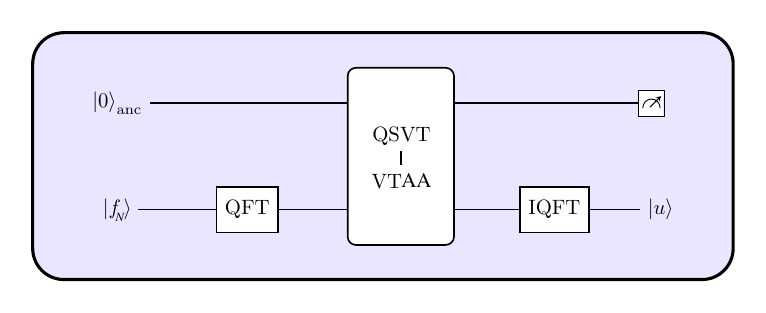}
	\caption{Circuit structure for the periodic biharmonic problem. the
    QSVT-VTAA linear-system solver prepares the augmented solution
    state, and the inverse QFT recovers the physical-space solution
    after extraction of the \(\widehat{\mathbf u}\)-component.}
	\label{fig:periodic}
\end{figure}

It remains to construct a block-encoding of the diagonal operator
\(\Lambda_N\). We first describe the one-dimensional frequency operator. Let
\(q_\ell=\ell-N/2\), \(\ell=0,1,\ldots,N-1\), and define
\[
K_N=\operatorname{diag}(q_0,q_1,\ldots,q_{N-1}).
\]
If \(\ell=\sum_{r=0}^{m-1}2^r\ell_r\) is the binary representation of
\(\ell\), then
\[
\operatorname{diag}(0,1,\ldots,2^m-1)
=
\sum_{r=0}^{m-1}2^r \Pi_r,
\]
where
\[
\Pi_r
=
I^{\otimes(m-r-1)}
\otimes \ket{1}\bra{1}
\otimes I^{\otimes r}.
\]
Using \(\ket{1}\bra{1}=(I-Z)/2\), we obtain
\begin{equation}\label{eq:centered-frequency-lcu}
	K_N
	=
	-\frac12 I
	-
	\sum_{r=0}^{m-1}2^{r-1} Z_r ,
\end{equation}
where \(Z_r\) denotes the Pauli-\(Z\) operator acting on the \(r\)-th qubit
under the above binary convention. Thus \(K_N\) is a linear combination of
Pauli-\(Z\) operators and the identity. 

The periodic Laplace operator in the Fourier basis is
\begin{equation}\label{eq:periodic-lambda}
	\Lambda_N
	=
	\pi^2
	\sum_{\ell=1}^d
	I^{\otimes(\ell-1)}
	\otimes K_N^2
	\otimes I^{\otimes(d-\ell)} .
\end{equation}
A block-encoding of \(K_N^2\) is obtained by composing two block-encodings of
\(K_N\). 
The augmented matrix can be decomposed as
\begin{equation}\label{eq:block-encoding-expand}
P_N
=
I\otimes\Lambda_N
-\frac12 X\otimes I
-\frac{\mathrm i}{2}Y\otimes I.
\end{equation}

\begin{lemma}[Block-encoding of the periodic augmented matrix]
\label{lem:periodic-block-encoding}
The matrix \(P_N\) in \eqref{eq:expanded-system} admits an exact
block-encoding with normalization
\[
\alpha_P
=
\frac{\pi^2dN^2}{4}+1.
\]
Under a direct implementation of the directional SELECT operation,
one use of this block-encoding has gate cost
\[
T_{P_N}=O(d\log N).
\]
Moreover, its normalization overhead satisfies
\[
\beta_{P_N}
:=
\frac{\alpha_P}{\norm{P_N}}
=
O(1).
\]
\end{lemma}

\begin{proof}
From \eqref{eq:centered-frequency-lcu}, the LCU normalization of
\(K_N\) is
\[
\alpha_{K_N}
=
\frac12+\sum_{r=0}^{m-1}2^{r-1}
=
\frac N2
=
\norm{K_N},
\]
and its block-encoding has gate cost \(O(\log N)\). Composing two
copies gives a block-encoding of \(K_N^2\) with normalization
\(N^2/4\) and the same asymptotic gate cost.

Applying a direct directional SELECT operation to the \(d\) terms
in \eqref{eq:periodic-lambda} therefore gives
\[
\alpha_\Lambda
=
\frac{\pi^2dN^2}{4}
=
\norm{\Lambda_N},
\qquad
T_{\Lambda_N}=O(d\log N).
\]
Applying Lemma \ref{lem:augmented-block-encoding} to
\eqref{eq:block-encoding-expand} yields
\[
\alpha_P=\alpha_\Lambda+1,
\quad
T_{P_N}=T_{\Lambda_N}+O(1)
       =O(d\log N).
\]

Finally,
\[
\norm{\Lambda_N}
\leq
\norm{P_N}
\leq
\norm{\Lambda_N}+1
=
\alpha_P.
\]
Hence
\[
1
\leq
\beta_{P_N}
=
\frac{\alpha_P}{\norm{P_N}}
\leq
1+\frac{1}{\norm{\Lambda_N}}
=
O(1).
\]
\end{proof}


\subsection{Complexity}
We now estimate the complexity of the periodic Fourier method.  For the augmented output, define
\[
\rho_N:=\frac{\norm{\widehat{\mathbf v}}}{\norm{\widehat{\mathbf u}}}
=\frac{\norm{\mathbf v}}{\norm{\mathbf u}},
\]
where the equality follows from the unitarity of the Fourier transform.
\begin{theorem}[Periodic boundary conditions]
\label{thm:periodic-complexity}
Consider the periodic biharmonic equation
\eqref{eq:periodic-biharmonic} on the mean-zero subspace,
and let the Fourier grid contain \(N\) points in each spatial
direction. Under the input model of Assumption~\ref{ass:input-output},
the QSVT-VTAA Fourier method prepares, up to state error
\(\varepsilon\), a normalized augmented state proportional to
$
\bigl(\widehat{\mathbf u},\widehat{\mathbf v}\bigr)
$
with gate complexity
\[
\Oht\!\left(d^2N^2\log\frac{1}{\varepsilon}\right).
\]
After amplitude-amplified extraction of the
\(\widehat{\mathbf u}\)-component and application of the inverse
quantum Fourier transform, it prepares \(\ket{\mathbf u}\) with
gate complexity
\[
\Oht\!\left(
d^2N^2\log\frac{1}{\varepsilon}\sqrt{1+\rho_N^2}
\right),
\qquad
\rho_N:=\frac{\norm{\widehat{\mathbf v}}}
{\norm{\widehat{\mathbf u}}}.
\]
Here the \(\Oht\)-notation suppresses polylogarithmic factors in
\(d\), \(N\).
\end{theorem}

\begin{proof}
The block-triangular forms of \(P_N\) and \(P_N^{-1}\), together
with \(\lambda_{\min}(\Lambda_N)=\pi^2\) and
\(\lambda_{\max}(\Lambda_N)=\Theta(dN^2)\), imply
\[
\kappa(P_N)=\Theta(dN^2).
\]

By Lemma \ref{lem:periodic-block-encoding},
\[
\beta_{P_N}=O(1),
\qquad
T_{P_N}=O(d\log N).
\]
Therefore, Lemma \ref{thm:qsvt-vtaa} gives
\[
\Oht\!\left(
\beta_{P_N}\kappa(P_N)T_{P_N}\log\frac{1}{\varepsilon}
\right)
=
\Oht(d^2N^2\log\frac{1}{\varepsilon})
\]
gates for preparing the normalized augmented solution state.

The probability of selecting the
\(\widehat{\mathbf u}\)-component is
\[
\frac{\norm{\widehat{\mathbf u}}^2}
{\norm{\widehat{\mathbf u}}^2+
 \norm{\widehat{\mathbf v}}^2}
=
\frac{1}{1+\rho_N^2}.
\]
Amplitude amplification therefore introduces an additional factor
\(O(\sqrt{1+\rho_N^2})\). The quantum Fourier transform and its inverse
have polylogarithmic gate cost and are absorbed into the
\(\Oht\)-notation.
\end{proof}

\begin{remark}[Behavior of \(\rho_N\) and discretization error]\label{rem:periodic-rho}
If the continuous solution is fixed while the grid is refined and \(u\in C^2(\overline\Omega)\) is periodic with \(\norm{u}_{L^2}\ne0\), then \(v=-\Delta u\in C(\overline\Omega)\) is fixed and \(\rho_N=O(1)\).  In this regime \Cref{thm:periodic-complexity} gives the fixed-grid state-preparation cost \(\Oht(dN^2\log(N/\varepsilon))\).  If one instead takes the worst case over arbitrary discrete right-hand sides, high-frequency modes give \(\rho_N=O(dN^2)\), and the conservative bound becomes \(\Oht(d^3N^4 \text{polylog}(N/\varepsilon))\).

If, in addition, the periodic solution is analytic in a complex neighborhood of the domain, the Fourier truncation error decays exponentially, so \(N=O(\log(1/\varepsilon))\) suffices for discretization error \(O(\varepsilon)\).  This yields \(\Oht(d^2 \text{polylog} (1/\varepsilon))\) in the fixed-smooth-solution regime. 
\end{remark}

\begin{remark}[State-output advantage regime]\label{rem:periodic-advantage}
A classical FFT-based spectral solver has near-linear cost \(\Oht(N^d)\) for producing the full vector of \(\Ndof=N^d\) grid values.  In the fixed-smooth-solution regime \(\rho_N=O(1)\), the quantum state-output cost is \(\Oht(d^2N^2)\), which is favorable in the power of \(N\) for integer dimensions \(d\ge3\).  In the arbitrary-discrete-data worst case \(\rho_N=O(dN^2)\), the bound is \(\Oht(d^3N^4)\), giving the threshold \(d\ge5\).  
These comparisons use the amplitude-encoding input and state-output model of Assumption~\ref{ass:input-output} and do not include full classical reconstruction from \(\ket{\mathbf u}\).
\end{remark}

\section{Simply supported boundary conditions}\label{sec:ss}

We next consider the biharmonic equation with simply supported boundary conditions,
\begin{equation}\label{eq:ss-pde}
\begin{cases}
    \Delta^2 u=f,\qquad x\in\Omega=(0,1)^d,\\
    u|_{\partial\Omega}=0,\qquad \Delta u|_{\partial\Omega}=0.
\end{cases}
\end{equation}
This case is important in plate theory. It is also algorithmically favorable because the equation decouples into two Dirichlet Poisson equations. Introducing 
\[
v=-\Delta u,
\]
the biharmonic equation can be reformulated as the decoupled system
\begin{align}\label{dls}
-\Delta v = f,\qquad -\Delta u = v.
\end{align}
Each Poisson equation is equipped with a Dirichlet boundary condition. Similar to the discussion in \Cref{sec:periodic}, we consider the following expanded linear system for the positive operator \(-\Delta\):
\[
\begin{bmatrix}
    -\Delta & -I\\
    0& -\Delta
\end{bmatrix}
\begin{bmatrix}
    u\\ v
\end{bmatrix}
=
\begin{bmatrix}
    0\\ f
\end{bmatrix}.
\]

\subsection{Finite difference method}\label{fdm}
We put $2^m + 1$ grid points on each coordinate, including the endpoints, with uniform mesh size. Let $N = 2^m-1$ be the number of interior grid points in each spatial direction. The mesh point is denoted by
\[
\mathbf x=(x^{(1)}_{j_1},x^{(2)}_{j_2},\ldots,x^{(d)}_{j_d}),
\qquad j_k=0,1,\ldots,N-1 .
\]
We apply the centered three-point formula to discretize the second-order partial derivatives on the $k$th coordinate of $u(\mathbf x)$. With the other variables fixed,
\begin{align*}
\frac{\partial^2u}{\partial x_k^2}(\mathbf x)
&=h^{-2}\Bigl(u(x_{j_k+1}^{(k)})-2u(x_{j_k}^{(k)})\\
&\hspace{2.0cm}+u(x_{j_k-1}^{(k)})\Bigr)+O(h^2).
\end{align*}
The corresponding $d$-dimensional discrete Laplacian is
\[
L_h^{(d)}
	=
	\sum_{\ell=1}^d
	I^{\otimes(\ell-1)}
	\otimes L_h^{(1)}
	\otimes I^{\otimes(d-\ell)},
\]
where $L_h^{(1)}$ is defined in Eq.~\eqref{laplace}. Therefore, the discrete linear system is 
\begin{align}\label{eq:dst-matrix}
\begin{bmatrix}
    L_h^{(d)} & -I\\
    0& L_h^{(d)}
\end{bmatrix}
\begin{bmatrix}
    \mathbf u\\ \mathbf v
\end{bmatrix}
=
\begin{bmatrix}
    0\\ \mathbf f_h
\end{bmatrix}.
\end{align}
Here \(\mathbf f_h\in\mathbb C^{\Ndof}\) denotes the vector of interior grid samples of the continuous source term \(f\).  Norms such as \(\norm{\mathbf f_h}\) in the complexity bounds refer to this discrete vector.
Denote
\begin{align}\label{eq:s-s-matrix}
S_N =
\begin{bmatrix}
    L_h^{(d)} & -I\\
    0& L_h^{(d)}
\end{bmatrix}.
\end{align}
The explicit block-encoding circuit for $L^{(d)}_h$ was recalled in \Cref{subsec:laplacian-be}. The decomposition of $P_N$ in Eq.~\eqref{eq:block-encoding-expand} can be applied to $S_N$. Furthermore, with the state-preparation circuits and LCU discussed previously, one can build the block-encoding of $S_N$. However, we do not use this procedure, since the block-encoding of $L^{(1)}_h$ is relatively more complicated than the block-encoding of the diagonal matrix. Instead, we exploit the fact that $L^{(1)}_h$ is diagonalized by the discrete sine transform (DST). 

The DST can be defined as follows: Let $ (x_1,x_2,\ldots,x_N)$ be a vector of length $N$. Its DST is 
\begin{align*}
X_k&=\sqrt{\frac{2}{N+1}}\sum_{n=1}^{N}x_n\sin\!\left(\frac{\pi nk}{N+1}\right),
\\&\hspace{1cm} k=1,2,\ldots,N .
\end{align*}
Let $D^{(1)}_h$ denote the diagonal matrix obtained by diagonalizing $L^{(1)}_h$. Its entries are
\begin{align*}
[D^{(1)}_h]_{kk}
&=\frac{4}{h^2}
\sin^2\!\left(\frac{k\pi}{2(N+1)}\right),\\
&\hspace{1cm} k=1,2,\ldots,N .
\end{align*}
To facilitate the computation of the DST and the construction of the matrix block-encoding on a quantum computer, we pad $D_{h}^{(1)}$ into the following $2^m \times 2^m$ matrix:
\[
D^{(1)} = \begin{bmatrix}
    0&\\&D_{h}^{(1)}
\end{bmatrix}_{2^m\times 2^m}.
\]
We then define the $d$-dimensional linear operator $D^{(d)}$ as
\[
D^{(d)} = \sum_{\ell=1}^d
	I^{\otimes(\ell-1)}
	\otimes D^{(1)}
	\otimes I^{\otimes(d-\ell)}.
\]
Thus the linear system in \eqref{eq:dst-matrix} comes to 
\begin{align}\label{eq:dst-system}
D_N \begin{bmatrix}
    \widehat{\mathbf u}\\\widehat{\mathbf v}
\end{bmatrix} = \begin{bmatrix}
    D^{(d)} & -I\\
    0& D^{(d)}
\end{bmatrix}\begin{bmatrix}
    \widehat{\mathbf u}\\\widehat{\mathbf v}
\end{bmatrix} = \begin{bmatrix}
    0\\\widehat{\mathbf f}_h
\end{bmatrix},
\end{align}
where $[0,\widehat{\mathbf f}_h]^{\mathsf T}$ is obtained by applying the DST to $[0,\mathbf f_h]^{\mathsf T}$. The physical-space vector $[\mathbf u,\mathbf v]^{\mathsf T}$ is recovered from $[\widehat{\mathbf u},\widehat{\mathbf v}]^{\mathsf T}$ by the inverse DST.

Notably, define the physical sine-transform subspace 
\begin{align*}
\mathcal H_{\rm phys} = \operatorname{span}\{\ket{k_1,\ldots,k_d}: 1\le k_{\ell}\le N,\\\quad \ell = 1,\ldots,d\}.
\end{align*}
The augmented matrix $D_N$ leaves 
$
\mathbb C^2\otimes \mathcal H_{\rm phys} 
$
 invariant. The transformed right-hand side belongs to this subspace and the restriction of the padded system $D_N$ is exactly the original DST-diagonalized finite-difference system. Moreover, the QSVT-VTAA inverse is applied on this invariant nonsingular subspace and the remaining padded models are not populated. 
\subsection{Block-encoding and circuit structure}
The overall circuit structure is shown in \cref{fig:qst-biharmonic}. The QSVT-VTAA procedure solves  system \eqref{eq:dst-system} for $[\widehat{\mathbf u},\widehat{\mathbf v}]^{\mathsf T}$. Conditioned on successful postselection of the ancilla registers, one can recover $[\mathbf u,\mathbf v]$ by the inverse DST.
\begin{figure}[h]
    \centering
    \maybeincludegraphics[width=1.0\linewidth]{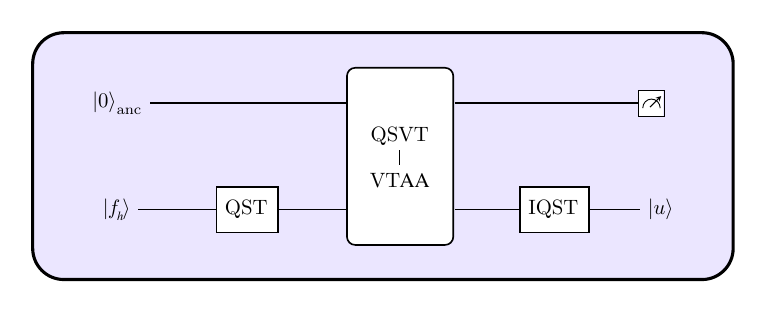}
    \caption{Circuit structure for the simply supported biharmonic problem.}
    \label{fig:qst-biharmonic}
\end{figure}

It remains to construct the quantum sine transform and the block-encoding of the linear operator in Eq.~\eqref{eq:dst-system}. We first focus on the implementation of the DST on a quantum computer. According to A. Klappenecker's work \cite{klappenecker2001discrete}, one can implement the DST on a quantum computer with complexity $O(\log^2 N)$. The explicit quantum circuit is shown in \cref{dst}. $S_1$ is defined as in Eq.~\eqref{eq:plus1}. $B$ is the one-qubit unitary
\[
B = \frac{1}{\sqrt{2}}
\begin{bmatrix}
1&i\\1&-i
\end{bmatrix}.
\]
\begin{figure}[h]
    \centering
    \maybeincludegraphics[width=1.0\linewidth]{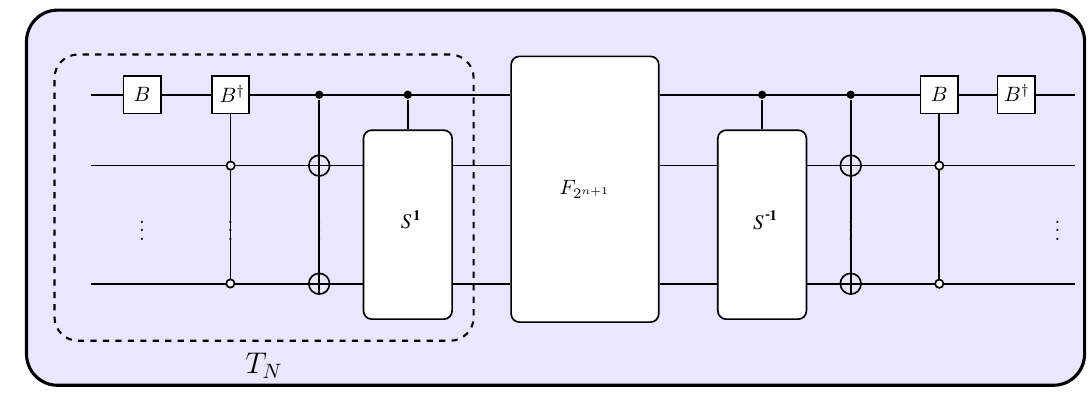}
    \caption{Explicit quantum circuit for the discrete sine transform (QST).}
    \label{dst}
\end{figure}
We briefly describe how this circuit implements the one-dimensional DST; the $d$-dimensional transform is obtained by applying it tensor-product-wise over the coordinates. In the DST procedure, we use $m+1$ qubits, so the total dimension is $M = 2^{m+1}=2(N+1)$. Let $N = 2^m-1$. We denote the computational basis states of these $m+1$ qubits as $|b, x\rangle$, where $b \in \{0,1\}$ is the most significant bit and $x \in \{0,1,\ldots,N\}$ labels the remaining $m$ qubits.

The discrete data to be transformed are the entries of $\mathbf f_h$. Assume that we prepare the initial state:
\[
\ket{\psi_{\rm in}} = \frac{1}{\|\mathbf f_h\|}\sum_{x=1}^{N} (\mathbf f_h)_x \ket{1,x}.
\]
(The integer index for $\ket{1,x}$ is $(N+1)+x$, and the amplitude of $\ket{1,0}$ remains zero.)

A direct calculation gives
\[
\begin{aligned}
T_N^\dagger F_{2(N+1)} T_N \ket{\psi_{\rm in}}
&= \ket{\psi_{\rm out}}\\
&= \frac{1}{\|\mathbf f_h\|}\sum_{k=1}^{N} i(\widehat{\mathbf f}_h)_k \ket{1,k}.
\end{aligned}
\]
Here $T_N$ is defined in \cref{dst}, 
\(F_{2^{n+1}}\) denotes the \(2^{n+1}\)-dimensional quantum Fourier
	transform,
and $(\widehat{\mathbf f}_h)_k$ is the DST coefficient of the discrete source vector $\mathbf f_h$. The global phase $i$ can be removed by a phase gate, such as $S^\dagger$, on the ancilla qubit. We then view $\ket{\psi_{\rm out}}$ as a $2^m$-dimensional vector with zero amplitude on $\ket{1,0}$, matching the padded diagonal matrix $D_h^{(1)}$. Thus the padded system is equivalent to the original system together with the trivial equation $0\cdot x=0$. For the $d$-dimensional case, one prepares the the amplitude state associated with the full grid vector $\mathbf f_h$ and applies the quantum sine transform on every coordinate.

Next we focus on designing the block-encoding of $D^{(1)}$. A recent work by Zecchi et al.~\cite{zecchi2026blockencodingsparsematrices} gives explicit quantum circuits to block-encode the diagonal matrix $V(\omega) = \operatorname{diag}(1,e^{i\omega},\cdots,e^{iN\omega})$. We can simply choose $\omega = \frac{\pi}{2(N+1)}$ and use the LCU algorithm on the block-encoding of $V(\omega)$ and $V(-\omega)$ to construct the block-encoding of
\[
\widehat{D} = \operatorname{diag}\!\left(0,\sin\frac{\pi}{2(N+1)},\ldots,\sin\frac{\pi N}{2(N+1)}\right).
\]
The explicit quantum circuits are shown in \cref{hatd} and \cref{fig:vomega}. The gate $P(\omega)$ is the phase gate
\[
P(\omega)=
\begin{bmatrix}
1&0\\
0&e^{i\omega}
\end{bmatrix}.
\]
\begin{figure}[h]
    \centering
    \maybeincludegraphics[width=1\linewidth]{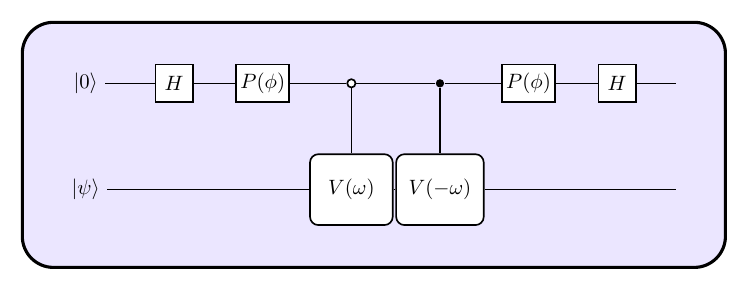}
    \caption{Block-encoding of $\widehat D$.}
    \label{hatd}
\end{figure}
\begin{figure}[h]
    \centering
    \maybeincludegraphics[width=1\linewidth]{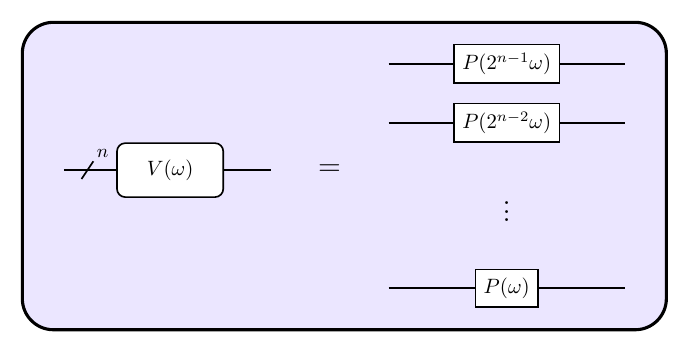}
    \caption{Block-encoding of $V(\omega)$.}
    \label{fig:vomega}
\end{figure}

Then $D^{(1)} = (4/h^2)\widehat{D}^2$ and $D^{(d)} = \sum_{\ell=1}^d
I^{\otimes(\ell-1)}
\otimes D^{(1)}
\otimes I^{\otimes(d-\ell)}$. Thus the block-encoding of $D^{(1)}$ can be efficiently built by matrix multiplication and LCU algorithms. One can then use the same decomposition in Eq.~\eqref{eq:block-encoding-expand} and implement the block-encoding for $D_N$. The resulting LCU construction is identical in form and is omitted for brevity.

\begin{lemma}[Block-encoding of the simply supported augmented matrix]
\label{lem:sim-sup-block-encoding}
The matrix \(D_N\) in \eqref{eq:dst-system} admits an exact
block-encoding with normalization
\[
\alpha_{D_N}
=
4d(N+1)^2 + 1.
\]
Under a direct implementation of the directional SELECT operation,
one use of this block-encoding has gate cost
\[
T_{D_N}=O(d\log N).
\]
Moreover, its normalization overhead satisfies
\[
\beta_{D_N}
:=
\frac{\alpha_{D_N}}{\norm{D_N}}
=
O(1).
\]
\end{lemma}

\begin{proof}
    According to \Cref{fig:vomega}, the LCU normalization of $D^{(d)}$ is
    $$
    \alpha_{D^{(d)}} = 4d(N + 1)^2
    $$
    and its block-encoding has gate cost $O(d\log N)$.
    Applying Lemma \ref{lem:augmented-block-encoding} to $D_N$ in system 
\eqref{eq:dst-system} yields
\begin{align*}
&\alpha_{D_N} = 4d(N+1)^2 + 1\\ &T_{D_N} = T_{D^{(d)}} + O(1) = O(d\log N).
\end{align*}

Finally, 
$$
1\leq\beta_{D_N} = \frac{4d(N+1)^2 + 1}{4d(N+1)^2\sin^2\frac{\pi N}{2(N+1)}}\leq 3 = O(1).
$$
\end{proof}

\subsection{Complexity}
We now estimate the complexity of the DST-based finite-difference method.  For the augmented output, define
\[
\rho_N:=\frac{\norm{\mathbf v}}{\norm{\mathbf u}} .
\]

\begin{theorem}[Simply supported boundary conditions]\label{thm:ss-complexity}
Consider the simply supported biharmonic equation~\eqref{eq:ss-pde} on a fixed uniform mesh with \(N=2^m-1\) interior points per spatial direction.  The QSVT-VTAA-based finite-difference method prepares, up to state error $\varepsilon$, a normalized augmented state proportional to \((\mathbf u,\mathbf v)\) with gate complexity
\[
\Oht\!\left(dN^2\log\frac{1}{\varepsilon}\right).
\]
After amplitude-amplified extraction of the solution component, it prepares \(\ket{\mathbf u}\) with gate complexity
\[
\Oht\!\left(dN^2\log\frac{1}{\varepsilon}\sqrt{1+\rho_N^2}
\right).
\]
\end{theorem}

\begin{proof}
On the non-padded interior subspace, \(D^{(d)}\) is invertible and satisfies \(\alpha_{D^{(d)}}=O(dN^2)\) and \(\norm{D^{(d)}}=O(dN^2)\).  The augmented matrix \(D_N\) has the same block form as \(P_N\). It is worth noticing that when we consider the matrix $D^{(d)}$ on the non-padded subspace, its maximum and minimum singular values have the following estimates:
\begin{align*}
    \sigma_{min}(D^{(d)}) = \Theta(d),\qquad \sigma_{max} = \Theta(dN^2).
\end{align*}
Hence \(\kappa(D_N)=O(N^2)\). 
By Lemma \ref{lem:sim-sup-block-encoding},
\[
\beta_{D_N}=O(1),
\qquad
T_{D_N}=O(d\log N).
\]
Therefore, Lemma \ref{thm:qsvt-vtaa} gives
\[
\Oht\!\left(
\beta_{D_N}\kappa(D_N)T_{D_N}\log\frac{1}{\varepsilon}
\right)
=
\Oht(dN^2\log\frac{1}{\varepsilon})
\]
gates for preparing the normalized augmented solution state.

The probability of selecting the
\(\mathbf u\)-component is
\[
\frac{\norm{\mathbf u}^2}
{\norm{\mathbf u}^2+
 \norm{\mathbf v}^2}
=
\frac{1}{1+\rho_N^2}.
\]
Amplitude amplification therefore introduces an additional factor
\(O(\sqrt{1+\rho_N^2})\). The quantum discrete sine transform and its inverse
have polylogarithmic gate cost and are absorbed into the
\(\Oht\)-notation.
\end{proof}

\begin{remark}[Behavior of \(\rho_N\) and discretization error]\label{rem:ss-rho}
For a fixed continuous solution with \(u\in C^2(\overline\Omega)\), compatible simply supported boundary data, and \(\norm{u}_{L^2}\ne0\), the auxiliary field \(v=-\Delta u\) is fixed and \(\rho_N=O(1)\) as \(h\to0\).  In the discrete worst case, the largest eigenvalue of the Dirichlet Laplacian is \(O(N^2)\), so \(\rho_N=O(N^2)\).

If the exact solution satisfies \(u\in C^6(\overline\Omega)\) and the boundary data satisfy the compatibility conditions needed for the standard second-order finite-difference error estimate, then the discretization error is \(O(h^2)\).  Taking \(N=O(\varepsilon^{-1/2})\) in \Cref{thm:ss-complexity} gives \(\Oht(d\varepsilon^{-1})\) in the fixed-smooth-solution regime and \(\Oht(d\varepsilon^{-2})\) in the discrete worst case.
\end{remark}

\begin{remark}[State-output advantage regime]\label{rem:ss-advantage}
Classical direct-transform or multigrid-type solvers have at least near-linear cost \(\Oht(N^d)\) in the number of unknowns when the full grid vector is produced.  In the fixed-smooth-solution regime \(\rho_N=O(1)\), \Cref{thm:ss-complexity} gives \(\Oht(dN^2)\), so the state-output comparison is favorable in the power of \(N\) for integer dimensions \(d\ge3\).  In the arbitrary-discrete-data worst case \(\rho_N=O(N^2)\), the bound \(\Oht(dN^4)\) gives the threshold \(d\ge5\).  These are state-output comparisons and do not include full classical reconstruction from \(\ket{\mathbf u}\).
\end{remark}

\begin{remark}[Nonhomogeneous boundary conditions]\label{rem:ss-nonhomogeneous}
The same transform-based construction applies to inhomogeneous simply supported data after moving the boundary contributions of \(u\) and \(v=-\Delta u\) to the right-hand side.  The matrix diagonalized by the DST is unchanged; only the prepared source state is modified.
\end{remark}

\section{Dirichlet--Neumann boundary conditions}\label{sec:dn}
In this section, we consider the biharmonic equation with general Dirichlet-Neumann boundary conditions
\begin{align}\label{eq:dn-biharmonic}
\begin{cases}
\Delta^{2}u = f,\ x\in[-1,1]^d,\\
u\vert_{\partial\Omega} = g_1,\ \frac{\partial u}{\partial n} = g_2.
\end{cases}
\end{align}
Unlike in the periodic and simply supported cases, Eq.~\eqref{eq:dn-biharmonic} does not split into two decoupled Poisson equations with known boundary data for the auxiliary variable. We first use the finite difference method to directly discretize the biharmonic operator $\Delta^2$.
\subsection{Finite difference method}
We first focus on the two-dimensional case, which has been fully discussed by P.~E.~Bj{\o}rstad in his work \cite{bjorstad1983fast}. Bj{\o}rstad's work employs a 13-point stencil to discretize the two-dimensional biharmonic operator. The resulting discretized matrix operator is formulated as a combination of several discrete Laplace operators $L^{(1)}_h$, which was defined in Eq.~\eqref{laplace}, and correction terms, and the system is subsequently solved using the Fast Sine Transform (FST) classically.  Inspired by this approach, we propose a quantum algorithm for solving $d$-dimensional biharmonic equations with general boundary conditions on a quantum computer. We begin by a brief review about the methodology for treating discrete operators as presented in Bj{\o}rstad's work.

Let $N = 2^n$ and set a uniform mesh with mesh size $h_x = h_y = \frac{2}{2^n+1}$ on square $[-1,1]^2$. Consider the following 13-point discretization to approximate the biharmonic operator. 
\begin{align}\label{eq:13stencil}
\Delta_{13}^{2} u
\;&\equiv\;
\frac{1}{h^{4}}
\begin{bmatrix}
  &   & 1   &   &   \\
  & 2 & -8  & 2 &   \\
1 & -8 & 20 & -8 & 1 \\
  & 2 & -8  & 2 &   \\
  &   & 1   &   &   
\end{bmatrix}u
\\&=
\frac{h^{2}}{6}
\left(
D_{1}^{6}
+
D_{1}^{4} D_{2}^{2}
+
D_{1}^{2} D_{2}^{4}
+
D_{2}^{6}
\right) u \notag \\
&\hspace{1cm}+\Delta^{2} u
+
\mathcal{O}(h^{4}) .
\end{align}
$D_{k}^{l}$ denotes the $l^{\text{th}}$-order partial derivative in coordinate direction $k$, with $k=1,2$. 

This stencil is only well defined for gridpoints $P$ having all their (nearest) neighbors in the interior of $\Omega$. To approximate the $4^{\text{th}}$-order derivative on points with a neighbor $Q\in\partial\Omega$, we use a one-dimensional five-point stencil in the normal direction and eliminate the ghost point through the boundary condition. Taking $\partial^4/\partial x^4$ as an example, for point $P$ next to the left boundary, denoted as $u_{1,y}$, consider the central scheme to approximate the Neumann boundary condition: 
\[
\frac{\partial u}{\partial x} = \frac{u_{1,y} - u_{-1,y}}{2h} + O(h^{2}).
\]
This gives an $O(h^{3})$ approximation of $u_{-1,y}$. Thus we have the following stencil for point $P$ only next to the left boundary:
\begin{align}
\Delta^2_q(P)& \equiv \frac{1}{h^4} 
\begin{bmatrix}
    & 1   &   &   \\
  2 & -8  & 2 &   \\
 -8 & 21 & -8 & 1 \\
   2 & -8  & 2 &   \\
     & 1   &   &   
\end{bmatrix}
u(P) + 2\frac{u_n(Q)}{h^3}\notag \\
& = \Delta^2 u(P) + O(h^{-1}).
\end{align}
This boundary stencil is not a pointwise consistent approximation of the fourth-order derivative near the boundary. However, Bramble \cite{bramble1966second} proved that in the 2D setting, although such discretization scheme suffers from local inconsistency at specific nodes, the resulting numerical solution still exhibits global second-order convergence. The $L_2$ norm of error $\Vert u-u_{h}\Vert = O(h^{2})$. 

Writing the scaled finite-difference matrix for this boundary treatment as $A_h^{\rm rough}$, the discrete problem takes the form
\[
A_h^{\rm rough}\,\mathbf u_h=\mathbf f_h .
\]
Define the boundary selector
\[
U=[e_1,e_N],
\]
where $e_i$ is the $i$th column of the $N\times N$ identity matrix.  The scaled matrix $A_h^{\rm rough}$ can be reformulated as
\begin{align*}
A_h^{\rm rough}&=\bigl(I\otimes L_h^{(1)}+L_h^{(1)}\otimes I\bigr)^2\\
&\quad +2h^{-4}UU^{\mathsf T}\otimes I+2h^{-4}I\otimes UU^{\mathsf T}.
\end{align*}
Equivalently,
\begin{align}
A_h^{\rm rough}&=\bigl((L_h^{(1)})^2+2h^{-4}UU^{\mathsf T}\bigr)\otimes I\notag\\
&\quad +I\otimes\bigl((L_h^{(1)})^2+2h^{-4}UU^{\mathsf T}\bigr)\notag\\
&\quad +2L_h^{(1)}\otimes L_h^{(1)} .
\end{align}
The first two terms discretize the pure fourth derivatives, and the last term discretizes the mixed derivative.  The matrix $UU^{\mathsf T}$ modifies only the first and last interior grid points in one dimension.  A direct $d$-dimensional analogue is obtained by using \(\mathcal L_h^{(d)}=\sum_{\ell=1}^d I^{\otimes(\ell-1)}\otimes L_h^{(1)}\otimes I^{\otimes(d-\ell)}\):
\begin{align}
A_h^{\rm rough}&=\bigl(\mathcal L_h^{(d)}\bigr)^2
+2h^{-4}\sum_{\ell=1}^d B_\ell,\label{eq:d-direct-rough}\\
B_\ell&=I^{\otimes(\ell-1)}\otimes UU^{\mathsf T}\otimes I^{\otimes(d-\ell)},
\end{align}
For $d\geq3$, this locally inconsistent boundary treatment does not directly provide a second-order global error estimate.  We therefore replace the boundary treatment  by a consistent one-sided approximation of the fourth derivative while keeping the interior tensor-product structure unchanged.  At the matrix level, this amounts to changing only the boundary modification matrix $UU^{\mathsf T}$.

Consider the following second-order one-sided approximation near a boundary:
\begin{align}\label{eq:near-boundary}
u^{(4)}(h)&=\frac{1}{h^4}\Bigl(
-\frac{113}{12}u_0+16u_1-9u_2 \notag\\
&\qquad +\frac83u_3-\frac14u_4-5h\,u'(0)
\Bigr)+O(h^2).
\end{align}
Denote the new discretization matrix by \(\widetilde A\).  With
\[
\mathcal L_h^{(d)}=
\sum_{\ell=1}^d I^{\otimes(\ell-1)}\otimes L_h^{(1)}\otimes I^{\otimes(d-\ell)},
\]
it can be written compactly as
\begin{align}\label{eq:atilde-def}
\widetilde A
&=\bigl(\mathcal L_h^{(d)}\bigr)^2
+\sum_{\ell=1}^d I^{\otimes(\ell-1)}\otimes M\otimes I^{\otimes(d-\ell)} .
\end{align}
Here
\begin{align}\label{eq:M-matrix}
M=\frac{1}{h^4}
\begin{bmatrix}
11&-5&5/3&-1/4&0&\cdots&0\\
0&0&0&0&0&\cdots&0\\
\vdots& & & & &\ddots&\vdots\\
0&\cdots&0&-1/4&5/3&-5&11
\end{bmatrix}.
\end{align}
After moving the known boundary contributions to the right-hand side, the discrete problem is
\begin{align}\label{eq:dn-system}
\widetilde A\,\mathbf u_h=\widetilde{\mathbf f}_h,
\end{align}
where \(\widetilde{\mathbf f}_h\) depends on \(f\), \(g_1\), and \(g_2\).  The scheme is consistent with the original PDE on every grid point.
\subsection{Convergence analysis}
We record the stability argument used for the direct boundary-corrected discretization.  The proof is included because this discretization is not simply the square of a symmetric Dirichlet Laplacian after the boundary correction is inserted.

\begin{lemma}\label{lem:singular}
Let \(B\in\mathbb R^{N\times N}\) and \(H(B)=(B+B^{\mathsf T})/2\).  If \(H(B)\succeq \alpha I\) with \(\alpha>0\), then \(B\) is invertible and \(\sigma_{\min}(B)\ge \alpha\).
\end{lemma}
\begin{proof}
For any nonzero real vector \(x\),
\[
\alpha\norm{x}_2^2\le x^{\mathsf T}H(B)x=x^{\mathsf T}Bx\le \norm{Bx}_2\norm{x}_2.
\]
Dividing by \(\norm{x}_2\) gives \(\norm{Bx}_2\ge \alpha\norm{x}_2\), and hence the claimed singular-value bound.
\end{proof}

\begin{lemma}\label{lem:stable}
Let \(M\) be the boundary-correction matrix in \eqref{eq:M-matrix}.  Let \(T_N=h^2L_h^{(1)}\) and \(M_0=h^4M\) be the corresponding dimensionless matrices.  There is a constant \(c>0\), independent of \(N\), such that
\[
H\!\left(T_N^2+M_0\right)\succeq c T_N^2 .
\]
Consequently, on a fixed interval there is a constant \(c_0>0\), independent of \(h\), such that
\[
H\!\left((L_h^{(1)})^2+M\right)\succeq c_0 I .
\]
\end{lemma}
\begin{proof}
The proof is local at the two boundary layers.  Write \(\mathbf d=-T_N\mathbf v\).  The interior contribution is \(v^{\mathsf T}T_N^2v=\norm{\mathbf d}_2^2\).  The symmetric boundary correction decomposes into left and right contributions supported on the first and last four grid points.  For the left boundary, introduce \(d_0=2v_1\) and \(d^{(L)}=(d_0,d_1,d_2,d_3)^{\mathsf T}\).  The boundary values \((v_1,v_2,v_3,v_4)^{\mathsf T}\) are related to \(d^{(L)}\) by the constant matrix
\[
W=\begin{bmatrix}
1/2&0&0&0\\
1&1&0&0\\
3/2&2&1&0\\
2&3&2&1
\end{bmatrix}.
\]
The left boundary quadratic form is \((d^{(L)})^{\mathsf T}B d^{(L)}\), where \(B=W^{\mathsf T}Z_{\rm left}W\) and
\[
Z_{\rm left}=\begin{bmatrix}
11&-5/2&5/6&-1/8\\
-5/2&0&0&0\\
5/6&0&0&0\\
-1/8&0&0&0
\end{bmatrix}.
\]
The right boundary has the same form after reversing the order.  The boundary-layer energy is controlled by
\begin{align*}
E&=\operatorname{diag}(0,1,1,1)+B,\\
\operatorname{spec}(E)
&=\left\{1,1,\frac{9}{8}\pm\frac{\sqrt{1082}}{48}\right\}.
\end{align*}
The smallest eigenvalue of this matrix is positive.  This gives a uniform lower bound for the two boundary-layer contributions relative to the neighboring components of \(d\).  Combining the two boundary layers with the positive interior contribution gives \(H(T_N^2+M_0)\succeq cT_N^2\).  Since the smallest eigenvalue of the scaled Dirichlet Laplacian \(L_h^{(1)}=h^{-2}T_N\) is bounded below on a fixed interval, the stated bound for \((L_h^{(1)})^2+M\) follows.
\end{proof}
\begin{theorem}[Second-order convergence of the direct discretization]\label{thm:dn-discretization}
Assume that the exact solution of \eqref{eq:dn-biharmonic} is sufficiently smooth, for instance \(u\in C^6(\overline\Omega)\), and that the Dirichlet and Neumann data satisfy the usual compatibility conditions at edges and corners.  Then the boundary-corrected finite-difference system \eqref{eq:dn-system} is stable and second-order consistent, and the discrete solution satisfies
\[
\frac{\norm{\mathbf u_I-\mathbf u_{h,I}}_2}{\norm{\mathbf u_I}}=O(h^2)
\]
up to the standard equivalence between the grid norm and the continuous \(L^2\) norm.
\end{theorem}
\begin{proof}
The interior stencil and the boundary formula \eqref{eq:near-boundary} are second-order consistent under the stated smoothness assumptions.  Lemma \ref{lem:singular} and Lemma \ref{lem:stable} imply a mesh-independent stability bound for the one-dimensional corrected fourth-derivative block and, by tensor-product summation, for the \(d\)-dimensional operator \(\widetilde A\).  Consistency plus stability gives the stated second-order convergence estimate.
\end{proof}

\subsection{Block-encoding and circuits}
The algorithm uses one QSVT-VTAA linear-system solve.  We therefore focus on constructing a block-encoding of $\widetilde A$.  The construction combines three ingredients: a block-encoding of $L_h^{(1)}$, a block-encoding of the boundary-correction matrix $M$, and an LCU construction for matrices of the form $T=Y\otimes I^{\otimes(d-1)}+\cdots+I^{\otimes(d-1)}\otimes Y$, assuming a block-encoding $U_Y$ of $Y$.

\begin{figure}[t]
    \centering
    \maybeincludegraphics[width=1.0\linewidth]{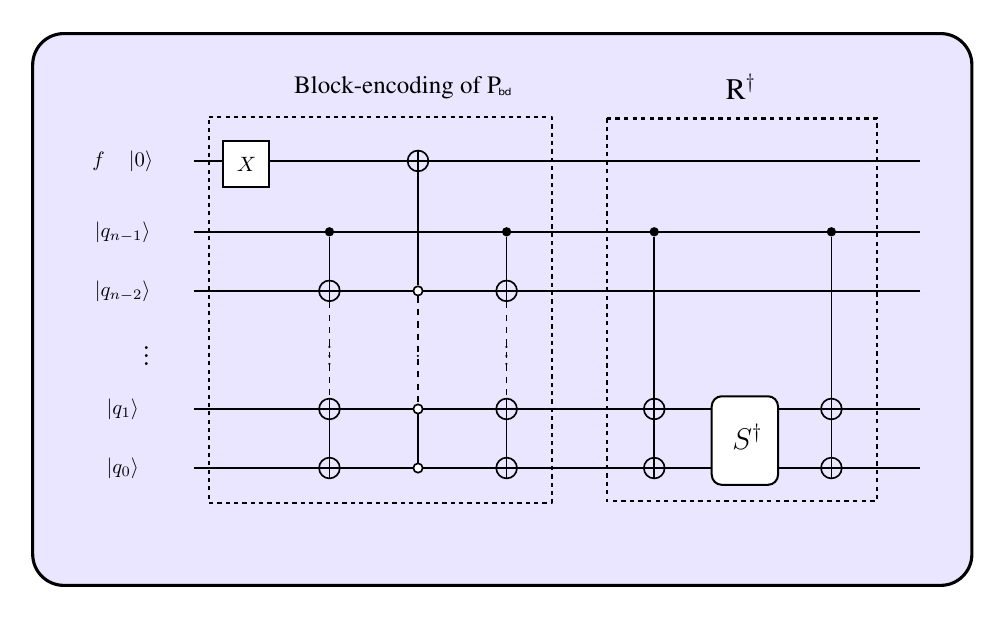}
    \caption{Optimal block-encoding of matrix $M$}
    \label{fig:be_m}
\end{figure}

The block-encoding of $L_h^{(1)}$ was given in Eq.~\eqref{eq:lap-shift-form}. We next present an exact, norm-matched block-encoding of the matrix $M$ with factor $\alpha = \|M\|_2$ and number of ancilla qubits $a = 1$.

Define the normalized two-qubit state associated with the integer-scaled first boundary row of \(h^4M\),
\[
\ket{u_M}=\frac{132\ket{00}-60\ket{01}+20\ket{10}-3\ket{11}}{\sqrt{21433}}.
\]
The integer scaling corresponds to multiplying the row \((11,-5,5/3,-1/4)\) by \(12\); the normalization constant and the global factor are absorbed into the block-encoding normalization \(\alpha_M=\norm{M}\).  By symmetry, \((X\otimes X)\ket{u_M}\) encodes the last nonzero boundary row.  Define
\begin{align*}
\ket{\phi_0}&=\ket{0^{n-2}}\ket{u_M},\\
\ket{\phi_1}&=\ket{1^{n-2}}(X\otimes X)\ket{u_M}.
\end{align*}
Let \(R\) be any unitary extension satisfying \(R\ket{0^n}=\ket{\phi_0}\) and \(R\ket{1^n}=\ket{\phi_1}\), and let
\[
P_{\rm bd}=\ket{0^n}\bra{0^n}+\ket{1^n}\bra{1^n}.
\]
Then the circuit in \Cref{fig:be_m}, with the boundary projector implemented by one block-encoding ancilla, has top-left block \(M/\alpha_M\).  Therefore it is an exact \((\alpha_M,1,0)\)-block-encoding of \(M\), with gate complexity \(O(\log N)\) plus the constant-size state preparation for \(\ket{u_M}\).  The state preparation can be implemented by the two controlled rotations shown in \Cref{s}, with
\begin{align*}
\theta&=2\arctan\frac{\sqrt{409}}{12\sqrt{146}},\\
\phi_0&=-2\arctan\frac{5}{11},\qquad
\phi_1=-2\arctan\frac{3}{20}.
\end{align*}

\begin{figure}[h]
    \centering
    \maybeincludegraphics[width=1\linewidth]{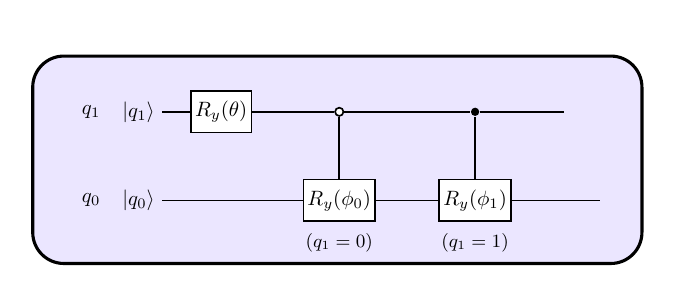}
    \caption{Explicit constant-size circuit for preparing the boundary-row state \(\ket{u_M}\).}
    \label{s}
\end{figure}

Finally, we give the block-encoding of matrix $T$ with structure 
\[
T = Y\otimes I^{\otimes(d-1)}+ \cdots+I^{\otimes (d-1)}\otimes Y
\]
with the block-encoding of $Y$, denoted as $U_Y$, prepared. Circuits are shown in \cref{p_be}.
\begin{figure}[h]
    \centering
    \maybeincludegraphics[width=1.0\linewidth]{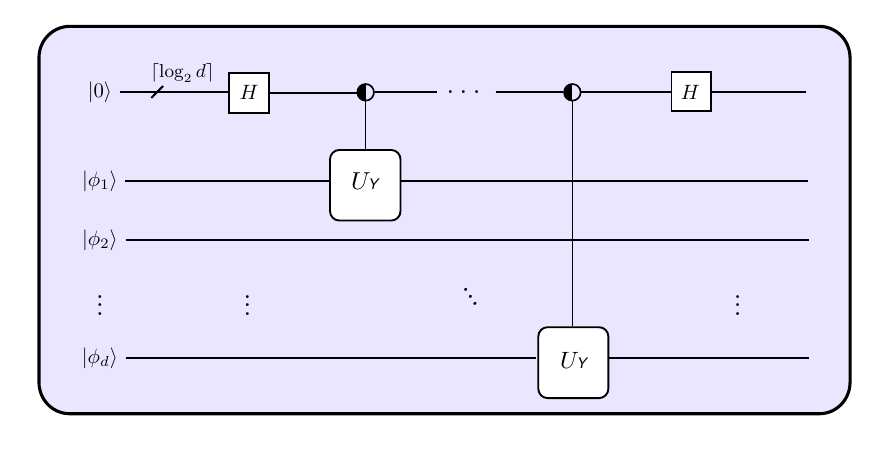}
    \caption{Block-encoding of matrix $T$}
    \label{p_be}
\end{figure}
Combining these ingredients by the block-encoding arithmetic gives a block-encoding of $\widetilde A$.

\begin{lemma}[Block-encoding of the direct discretization matrix]
\label{lem:d-d-block-encoding}
The matrix \(\widetilde A\) in \eqref{eq:atilde-def} admits an exact
block-encoding with normalization
\[
\alpha_{\widetilde A}
=
16d^2(N+1)^4 + \frac{\sqrt{21433}}{12}d(N+1)^4.
\]
Under a direct implementation of the directional SELECT operation,
one use of this block-encoding has gate cost
\[
T_{\widetilde A}=O(d\log N).
\]
Moreover, its normalization overhead satisfies
\[
\beta_{\widetilde A}
:=
\frac{\alpha_{\widetilde A}}{\norm{\widetilde A}}
=
O(1).
\]
\end{lemma}
\begin{proof}
The explicit block-encoding of $L^{(1)}_h$ in \eqref{eq:lap-shift-form} gives the factor $\alpha_{L^{(1)}_h}  = 4(N+1)^2$. Thus the LCU normalization of $L^{(d)}_h$ is $\alpha_{L^{(d)}_h} = 4d(N+1)^2$ and its block-encoding has gate cost $O(d\log N)$.  As shown in \Cref{fig:be_m}, the block-encoding factor $\alpha_M$ of matrix $M$ is optimal, i.e. $\alpha_M = \norm{M} = \frac{\sqrt{21433}}{12}(N+1)^4$ and its block-encoding has gate cost $O(d\log N)$. Consequently, LCU is used to construct the block-encoding of matrix $\widetilde{A}$ and the factor 
$$
\alpha_{\widetilde{A}} = 16d^2(N+1)^4 + \frac{\sqrt{21433}}{12}d(N+1)^4.
$$
The estimation of $\|\tilde{A}\|_2$ is given by
    \[
    \|\tilde{A}\|_2 = \frac{\|\tilde{A}\|_2 + \|\tilde{A}^T\|}{2}\geq \|H(\tilde{A})\|_2.
    \]
According to Lemma. \ref{lem:stable},
\begin{align*}
    H(\widetilde A) &= \sum\limits_{\ell = 1}^{d}H(L_{\ell}^{2} + M) + 2\sum\limits_{\ell<j}L_{\ell}L_{j}\\
    &\succeq c \sum_{\ell=1}^d L_{\ell}^2 + 2 \sum_{\ell < j} L_{\ell} L_j \\
&\succeq c \left( L_h^{(d)} \right)^2,
\end{align*}
with $0<c<1$. Given that the inequality $\norm{\widetilde A}\leq \alpha_{\widetilde A}$ always holds, we conclude that
  $$
  1\leq\beta_{\widetilde A}\leq\frac{1}{c}(16 + \frac{\sqrt{21433}}{12}) = O(1).
  $$
  
Regarding the gate complexity $T_{\widetilde{A}}$, based on our explicit construction of the block-encoding, it is evident that
\begin{align*}
    T_{\widetilde A} &= 2T_{L^{(d)}_h} + T_{M} + O(1) \\
    &= O(d\log N).
\end{align*}
\end{proof}

\subsection{Complexity}\label{subs:fd-complexity}
We now estimate the complexity of the direct finite-difference method for the Dirichlet--Neumann boundary problem.  The following theorem uses the QSVT-VTAA model in Lemma \ref{thm:qsvt-vtaa}.

\begin{theorem}[Direct Dirichlet--Neumann discretization]\label{thm:dn-complexity}
For the \(d\)-dimensional Dirichlet--Neumann biharmonic problem \eqref{eq:dn-biharmonic}, assume the regularity and compatibility hypotheses of \Cref{thm:dn-discretization}.  For a fixed grid with \(N\) interior points per spatial direction, the QSVT-VTAA-based direct finite-difference method prepares, up to state error $\varepsilon$, the normalized discrete solution state with worst-case gate complexity
\[
\widetilde O\!\left(dN^4\log\frac{1}{\varepsilon}\right),
\]
Since the discretization is second order, choosing \(N=O(\varepsilon^{-1/2})\) gives the end-to-end complexity
\[
\widetilde O\!\left(d\varepsilon^{-2}\right).
\]
\end{theorem}
\begin{proof}
We firstly derive the upper bound of $\kappa(\tilde{A})$. According to Lemma \ref{lem:singular} and Lemma \ref{lem:stable}:
    \[
    \sigma_{min}(\widetilde{A})\geq Cd\lambda_{min}^2(L^{(1)}_h) + \frac{d(d-1)}{2}\lambda_{min}^2(L_h^{(1)}).
    \]
    And triangular inequality gives
    \[
    \sigma_{max}(\widetilde{A}) = \Vert \widetilde{A}\Vert_2\leq d^2\lambda_{max}^2(L_{h}^{(1)}) + d\|M\|_2.
    \]
    Thus one can concludes
    \[
    \kappa(\widetilde{A})\leq O(N^4).
    \]

By Lemma \ref{lem:d-d-block-encoding},
\[
\beta_{\widetilde{A}}=O(1),
\qquad
T_{\widetilde{A}}=O(d\log N).
\]
Therefore, Lemma \ref{thm:qsvt-vtaa} gives
\[
\Oht\!\left(
\beta_{\widetilde{A}}\kappa(\widetilde{A})T_{\widetilde{A}}
\right)
=
\Oht(dN^4)
\]
gates for preparing the normalized solution state.

\end{proof}


\begin{remark}[Classical comparison and advantage regime]
\label{rem:dn-classical-comparison}
For the nonsymmetric boundary-corrected system, we use LSQR, or equivalently
CGLS applied implicitly to the normal equations, as an unpreconditioned
classical baseline. Denote \(\operatorname{nnz}(A)\)  the number of
nonzero entries of a matrix \(A\). Since
\(\operatorname{nnz}(\widetilde A)=O(d^2N^d)\) and
\(\kappa(\widetilde A)=O(N^4)\), the corresponding classical arithmetic
complexity is \(\Oht(d^2N^{d+4})\), whereas the quantum state-preparation cost
is \(\Oht(dN^4)\). Under the input-output model of
Assumption~\ref{ass:input-output}, the quantum scaling is therefore
asymptotically favorable for every \(d\ge1\) relative to this classical
baseline. It is noted that the present work does not analyze problem-dependent preconditioning, either classical
or quantum, and therefore does not claim an unconditional quantum advantage for
the boundary-corrected system.
\end{remark}

\section{Coupled Laplace formulation}\label{sec:coupled}
Based on the condition-number and error analysis in \Cref{sec:dn}, it is evident that the discrete matrices resulting from direct discretization of the biharmonic operator typically exhibit a large condition number scaling as $O(N^4)$. Consequently, to mitigate the ill-conditioning and reduce computational complexity, we introduce a splitting formulation. Originally developed for spectral collocation methods \cite{Heinrichs1991AST,legendre_colla}, this approach transforms the biharmonic equation into a system of two coupled Laplace equations by introducing an auxiliary variable $v$. Under this framework, both spectral collocation and finite difference methods can be employed to solve the biharmonic equation.

Consider the following PDE system:
\[
\begin{cases}
\Delta^{2}u = f, & x\in[-1,1]^d,\\
u\vert_{\partial\Omega} = g_1, & \partial u/\partial n = g_2.
\end{cases}
\]
Introduce the auxiliary variable $v=-\Delta u$. The coupled system is
\[
\begin{cases}
    -\Delta u = v,\qquad -\Delta v = f,\\
    u\vert_{\partial\Omega} = g_1,\qquad \partial u/\partial n = g_2.
\end{cases}
\]
Unlike the simply supported case, the boundary information for the auxiliary variable $v$ is not explicitly available. We therefore treat the boundary values of $v$ as additional unknowns and incorporate them into an augmented linear system.

\subsection{Discrete Model}
Assume that the domain $[-1,1]^d$ is discretized by a uniform grid with $N+2$ grid points in each coordinate direction, so that $N$ interior points are retained.  Let $L_h$ denote the interior discrete positive Laplacian.  Let $\mathbf u_I$ and $\mathbf v_I$ be the vectors of interior values of $u$ and $v$, and let $\mathbf v_B$ collect the boundary values of $v$ that enter the interior Laplacian.  The coupled finite-dimensional system is
\begin{align}\label{eq:split-linear}
K\mathbf x:=
\begin{bmatrix}
    L_h&-I&0\\
    0&L_h&A_{IB}\\
    A_{\rm apm}&0&0
\end{bmatrix}
\begin{bmatrix}
    \mathbf u_I\\ \mathbf v_I\\ \mathbf v_B
\end{bmatrix}
=
\begin{bmatrix}
    \mathbf g_1\\ \mathbf f_h\\ \mathbf g_2
\end{bmatrix}.
\end{align}
Here $A_{IB}$ represents the contribution of the boundary values of $v$ to the interior Laplace equations, while $A_{\rm apm}$ maps the interior values $\mathbf u_I$ to the discrete normal derivative on the boundary.  In the homogeneous clamped case, $\mathbf g_1$ and $\mathbf g_2$ vanish after the usual boundary contributions are moved to the right-hand side.

Consider the two-dimensional biharmonic equation with homogeneous boundary condition:
\[
\begin{cases}
\Delta^{2}u = f, & x\in[-1,1]^2,\\
u\vert_{\partial\Omega} = 0, & \partial u/\partial n = 0.
\end{cases}
\]
The five-point finite-difference stencil with $N+2$ points in each dimension is applied to discretize the Laplace operator. We assume that the vectorization of the grid values $u_{ij}$ and $v_{ij}$ follows tensor-product ordering.  Define the boundary vector \(\mathbf v_B\), collecting the four edge contributions to the interior Laplacian, by
\begin{align*}
\mathbf v_B&=
\begin{bmatrix}
\mathbf v_{x^-}\\ \mathbf v_{x^+}\\ \mathbf v_{y^-}\\ \mathbf v_{y^+}
\end{bmatrix},\\
\mathbf v_{x^-}&=\{v_{0,j}\},\quad
\mathbf v_{x^+}=\{v_{N+1,j}\},\\
\mathbf v_{y^-}&=\{v_{i,0}\},\quad
\mathbf v_{y^+}=\{v_{i,N+1}\}.
\end{align*}
Then the block \(A_{IB}\) is
\begin{align}\label{eq:a_ib}
A_{IB}&=-\frac{1}{h^2}
\begin{bmatrix}E_{x^-}&E_{x^+}&E_{y^-}&E_{y^+}\end{bmatrix},\notag\\
E_{x^-}&=e_1^x\otimes I,
\qquad E_{x^+}=e_N^x\otimes I,\notag\\
E_{y^-}&=I\otimes e_1^y,
\qquad E_{y^+}=I\otimes e_N^y .
\end{align}
Here $e_1^x=e_1^y=(1,0,\ldots,0)^{\mathsf T}$ and $e_N^x=e_N^y=(0,\ldots,0,1)^{\mathsf T}$.
In order to determine the matrix block $A_{\rm apm}$, we first give the discrete differentiation formula that approximates $\partial u/\partial n$ from the interior vector $\mathbf u_I$. Since the five-point stencil is a second-order approximation of the Laplace operator, it is sufficient to consider the one-sided approximation scheme with $3$ points. We take the mesh points on $x^{-}$ line as an example:
\[
\partial_n u|_{x=0}=-u_x(0,y_j).
\]
The second-order one-sided discretization is 
\begin{align*}
-u_x(0,y_j)
&=\frac{1}{2h}\bigl(3u(0,y_j)-4u(x_1,y_j)\\
&\hspace{2.6cm}+u(x_2,y_j)\bigr)+O(h^2).
\end{align*}
where $u(0,y_j)$ is supplied by the boundary condition. Thus we use the row vectors 
\begin{align*}
d_{x^-}^{\mathsf T}&=\frac{1}{2h}(-4,1,0,\ldots,0),\\
d_{x^+}^{\mathsf T}&=\frac{1}{2h}(0,\ldots,0,1,-4).
\end{align*}
and define the matrices
\begin{align*}
    A_{x^-} = d_{x^-}^{\mathsf T}\otimes I,\qquad  A_{x^+} = d_{x^+}^{\mathsf T}\otimes I.
\end{align*}
Similarly, define
\begin{align*}
    A_{y^-} = I\otimes d_{y^-}^{\mathsf T},\qquad A_{y^+} = I\otimes d_{y^+}^{\mathsf T}.
\end{align*}
Let
\begin{align*}
    A_{\rm apm} = \begin{bmatrix}
        A_{x^-}\\
        A_{x^{+}}\\
        A_{y^-}\\
        A_{y^+}
    \end{bmatrix},
\end{align*}
which gives the second-order approximation of the boundary normal derivative
\[
A_{\rm apm}\mathbf u_I = \partial_n \mathbf u_B + O(h^2).
\]
Evidently, the two-dimensional formulation can be straightforwardly generalized to the $d$-dimensional case. The homogeneous treatment extends to inhomogeneous boundary data by moving known boundary contributions to the right-hand side. We now turn our attention to the implementation of the block-encoding for this newly formed matrix.

\subsection{Block-encoding and circuits}\label{subsec:be-circuit}
Taking the two-dimensional discrete linear system from the previous subsection as a representative example, we now construct the block-encoding for the corresponding discrete model.

First, as established in the preceding discussion \eqref{eq:lap-shift-form}, the block-encoding of $L_h$ is straightforward to implement. We therefore focus our attention on the realization of the block-encodings for matrices $A_{\rm apm}$ and $A_{IB}$.

Consistent with Eq.~\eqref{eq:a_ib}, $A_{IB}$ is a real-valued matrix of size $N^2 \times 4N$ in two dimensions and $N^d\times 2dN^{d-1}$ in $d$ dimensions. Assume for simplicity that $N = 2^m$. To ensure compatibility with quantum encoding schemes, it is necessary to pad $A_{IB}$ into a square matrix acting on $md$ qubits. Denote $S = [A_{IB},\mathbf{0}]$. We reformulate $S$ in the following sense:
\begin{align}
S&=-h^{-2}\sum_{k=1}^d\left(B_{k,0}+B_{k,1}\right),\label{eq:s}\\
B_{k,0}&=\bra{2k-2}\otimes I^{\otimes(k-1)}\otimes\ket{0}\otimes I^{\otimes(d-k)},\\
B_{k,1}&=\bra{2k-1}\otimes I^{\otimes(k-1)}\otimes\ket{2^m-1}\otimes I^{\otimes(d-k)}.
\end{align}
The strategy is to first construct block-encodings for the $2d$ matrices of the form $\bra{\alpha}\otimes I^{\otimes(k-1)}\otimes\ket{\beta}\otimes I^{\otimes(d-k)}$, and then use LCU to sum them, thereby obtaining a block-encoding of $S$. We now derive the block-encoding of square matrix $G= \bra{\alpha}\otimes I^{\otimes k-1}\otimes\ket{\beta}$.

Consider registers $R_0,R_1,\ldots,R_{k-1}$, each consisting of $m$ qubits.  Acting on a computational basis state, the matrix $G$ satisfies
\begin{align*}
G\ket{x}_{R_0}\ket{s_1,\ldots,s_{k-1}}
=\delta_{x,\alpha}\ket{s_1,\ldots,s_{k-1},\beta}.
\end{align*}
Define two unitary operators $E_\alpha,\ T_{\alpha\to \beta}$ for fixed $\alpha,\beta$, which satisfy
\[
E_\alpha\ket{\alpha} = \ket{1\cdots1},\qquad T_{\alpha\to \beta}\ket{\alpha} = \ket{\beta}.
\]
Since $\alpha$ and $\beta$ are fixed known integers, these two operators can be implemented with $X$ gates on selected qubits.  The circuit in \cref{qcy} block-encodes $G$ with one ancilla qubit and normalization $\alpha_G=1=\norm{G}$. We can consequently implement the block-encoding of $S$ by applying LCU with $1+\lceil\log_2(2d)\rceil$ ancilla qubits and factor $\alpha_S=2d/h^2$.
\begin{figure}[h]
    \centering
    \maybeincludegraphics[width=1\linewidth]{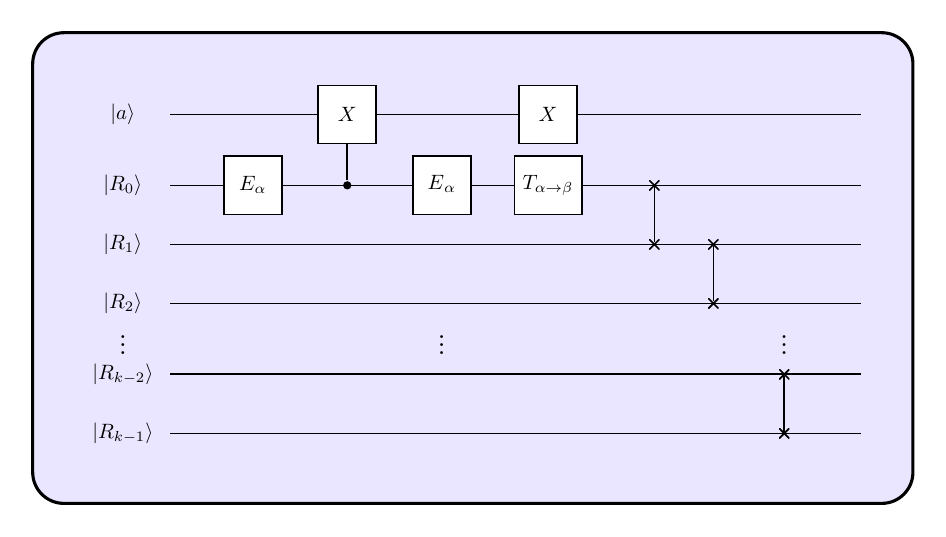}
    \caption{$U_G$: block-encoding of the matrix $G$.}
    \label{qcy}
\end{figure}

Regarding matrix $A_{\rm apm}$, it shares an analogous structure with $A_{IB}$; hence, the methodology for constructing the block-encoding of $A_{IB}$ is directly applicable to $A_{\rm apm}$. It should be noted that, unlike $A_{IB}$, $A_{\rm apm}$ requires padding with a bottom zero-block to be transformed into a square matrix. We omit further details here for the sake of brevity. Once the block components are available, another LCU step constructs a block-encoding of the full matrix $K$, analogously to Eq.~\eqref{eq:block-encoding-expand}.

\begin{lemma}[Block-encoding of the coupled augmented matrix]
\label{lem:couple-block-encoding}
The matrix \(K\) in \eqref{eq:split-linear} admits an exact
block-encoding with normalization
\[
\alpha_{K}
=
6dN^2 + 4dN + 1.
\]
Under a direct implementation of the directional SELECT operation,
one use of this block-encoding has gate cost
\[
T_{K}=O(d\log N).
\]
Moreover, its normalization overhead satisfies
\[
\beta_{K}
:=
\frac{\alpha_{K}}{\norm{{K}}}
=
O(1).
\]
\end{lemma}

\begin{proof}
    First, it is obvious that 
    $$
    \norm{K}\geq dN^2.
    $$
    As we discussed in \cref{subsec:be-circuit}, the block-encoding of matrix $B_{k,0}$ and $B_{k,1}$ is optimal, i.e.
    $$
    \alpha_B = \norm{B_{k,0}} = \norm{B_{k,1}}.
    $$
    According to the norm of tensor product matrix, it is evident that
    $$
    \norm{B_{k,0}} = \norm{B_{k,1}} = 1.
    $$
    From \eqref{eq:s}, the LCU normalization of $S$ is 
    $$
    \alpha_S = \alpha_{A_{IB}} =  2dN^2.
    $$
    With the similar analysis, we conclude that the LCU normalization of $A_{apm}$ is
    $$
    \alpha_{A_{apm}} = 4dN.
    $$
    In order to construct the block-encoding of the matrix $K$, the linear combination of unitaries (LCU) technique is deployed once more to assemble the constituent block matrices. As established in \Cref{sec:prelim}, the block-encoding coefficient for the block-diagonal matrices $L_h^{(d)}$ is $4dN^2$. Therefore, the total block-encoding coefficient of the matrix $K$ satisfies $\alpha_K = 6dN^2 + 4dN + 1$. Furthermore, given that the inequality $\alpha_K \geq \Vert{}K\Vert{}$ always holds and that $\Vert{}K\Vert{} \geq dN^2$, we conclude that 
    $$1\leq\beta_{K}\leq\frac{\alpha_K}{dN^2}\leq 7 = O(1).$$

    As for the gate complexity $T_K$, it readily follows from our explicit construction of the block-encoding that
    $$
    T_K = O(d\log N).
    $$
\end{proof}
\subsection{Discussion}
The coupled formulation has an explicit block-encoding.  Its QSVT-VTAA complexity is naturally expressed in terms of the condition number \(\kappa(K)\) of the boundary-augmented matrix.  In particular, if the one-use block-encoding cost is denoted by \(T_K\), then the QSVT-VTAA model gives a gate complexity of order
\[
\Oht\bigl(\kappa(K) T_K\log\frac{1}{\varepsilon}\bigr)
\]
for preparing, up to state error $\varepsilon$, the normalized augmented solution state, up to the normalization factors from the block-encoding. Furthermore, to extract solution $\ket{\mathbf u_I}$, another factor $\sqrt{1 + \rho_{N}^{2}}$ will be multiplied to the complexity estimate, where 
$$
\rho_N^2 = \frac{\norm{\mathbf v_I}^2 + \norm{\mathbf v_B}^2}{\norm{\mathbf u_{I}}^2}.
$$
Consequently, using \(T_K=O(d\log N)\) and \(\beta_K=O(1)\), the overall gate complexity for preparing the normalized physical solution state \(\ket{\mathbf u_I}\) is
\[
\Oht\!\left(
d\,\kappa(K)\sqrt{1+\rho_N^2}
\log\frac{1}{\varepsilon}
\right).
\]
A general singular-value estimate for \(K\) depends on the boundary discretization and dimension.  The numerical study in \Cref{subsec:num-coupled} is therefore included to quantify this conditioning behavior for the two-dimensional finite-difference test problem.

The same splitting idea can also be combined with spectral collocation or Galerkin spectral methods.  Classical spectral discretizations of biharmonic Dirichlet problems can be very ill-conditioned, and stabilized splitting formulations are known to improve the conditioning in Chebyshev or Legendre bases \cite{Heinrichs1991AST,legendre_colla}.  In a quantum algorithm, this improvement is valuable only if the boundary operators introduced by the splitting admit block-encodings with normalization and gate costs that do not offset the conditioning gain.

\section{Numerical experiments}\label{sec:numerics}

This section reports numerical experiments validating the proposed circuit constructions and the finite-difference discretizations used in the analysis. The first two experiments compare small QSVT circuit simulations with the corresponding classical discretizations and verify that the circuits implement the intended discrete solvers. The final coupled-Laplace experiment reports the discretization error and the empirical growth of the boundary-augmented condition number appearing in \Cref{sec:coupled}. The quantum circuits are simulated classically at small problem sizes using {\it UnitaryLab} \cite{UnitaryLab2026}.

It is noted that while asymptotic analysis uses QSVT–VTAA, the experiments implement simpler fixed-time QSVT inversion. This choice validates the core contribution—the block-encodings and discrete operators—since swapping solvers alters complexity but not the tested matrices.

\subsection{Periodic boundary conditions}\label{subsec:num-periodic}

We first consider the one-dimensional periodic biharmonic equation
\begin{equation}\label{eq:num-periodic}
\begin{cases}
u^{(4)}(x)=\pi^4\bigl(\sin(\pi x)+16\cos(2\pi x)\bigr),\\
u^{(k)}(-1)=u^{(k)}(1),\quad k=0,1,2,\ldots,
\end{cases}
\end{equation}
The exact solution is
\begin{equation}
    u(x)=\sin(\pi x)+\cos(2\pi x).
\end{equation}
The simulation uses $8$ collocation points. The result is shown in \Cref{fig:qsvtbi}; the quantum-circuit solution agrees with the classical spectral discretization at the same resolution, 
 confirming that the implemented circuit realizes the
intended discrete linear solve.

\begin{figure}[h]
    \centering
    \maybeincludegraphics[width=0.75\linewidth]{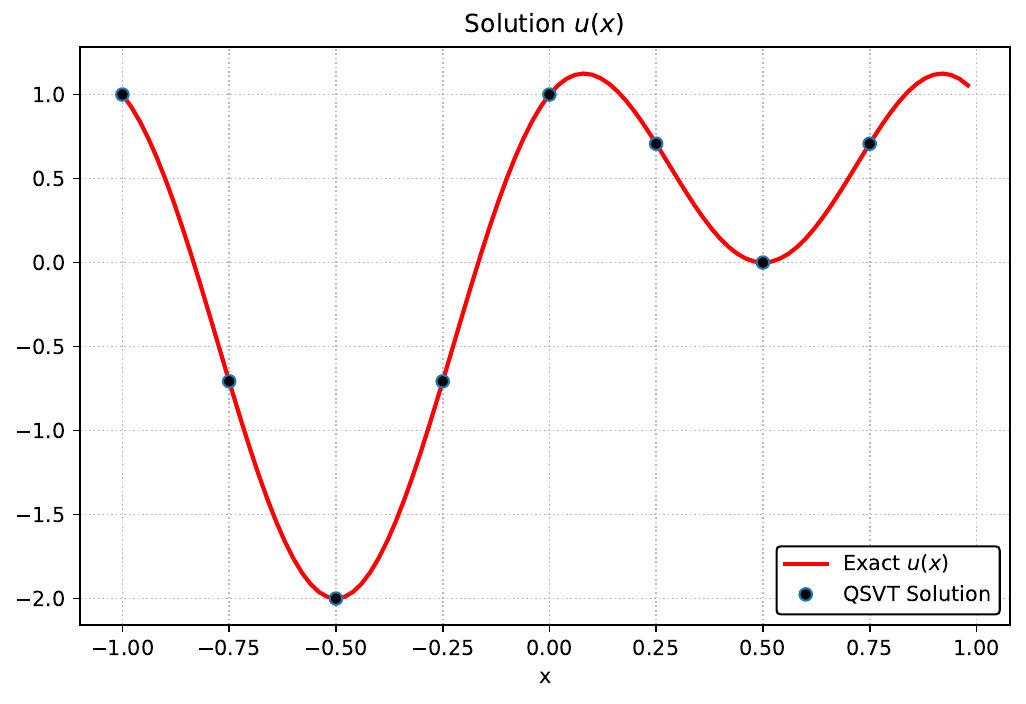}
    \caption{Periodic problem \eqref{eq:num-periodic}. Comparison between the exact solution, the classical spectral discretization, and the QSVT circuit simulation.}
    \label{fig:qsvtbi}
\end{figure}

\subsection{Simply supported boundary conditions}\label{subsec:num-ss}

We next solve
\begin{equation}\label{eq:num-ss}
\begin{cases}
u^{(4)}(x)=\pi^4\bigl(\sin(\pi x)+40.5\sin(3\pi x)\bigr),\\
u|_{\partial\Omega}=0,
\qquad \Delta u|_{\partial\Omega}=0,
\end{cases}
\end{equation}
The exact solution is
\begin{equation}
    u(x)=\sin(\pi x)+0.5\sin(3\pi x).
\end{equation}
Using $15$ interior points, the infinity-norm error of the quantum simulation is $3.298\times 10^{-2}$, while the corresponding classical finite-difference error is $3.286\times 10^{-2}$. The comparison is shown in \Cref{fig:dstbi}. 
This agreement confirms that the QSVT circuit realizes the intended
finite-difference linear solve at the tested resolution.
\begin{figure}[ht]
    \centering
    \maybeincludegraphics[width=0.75\linewidth]{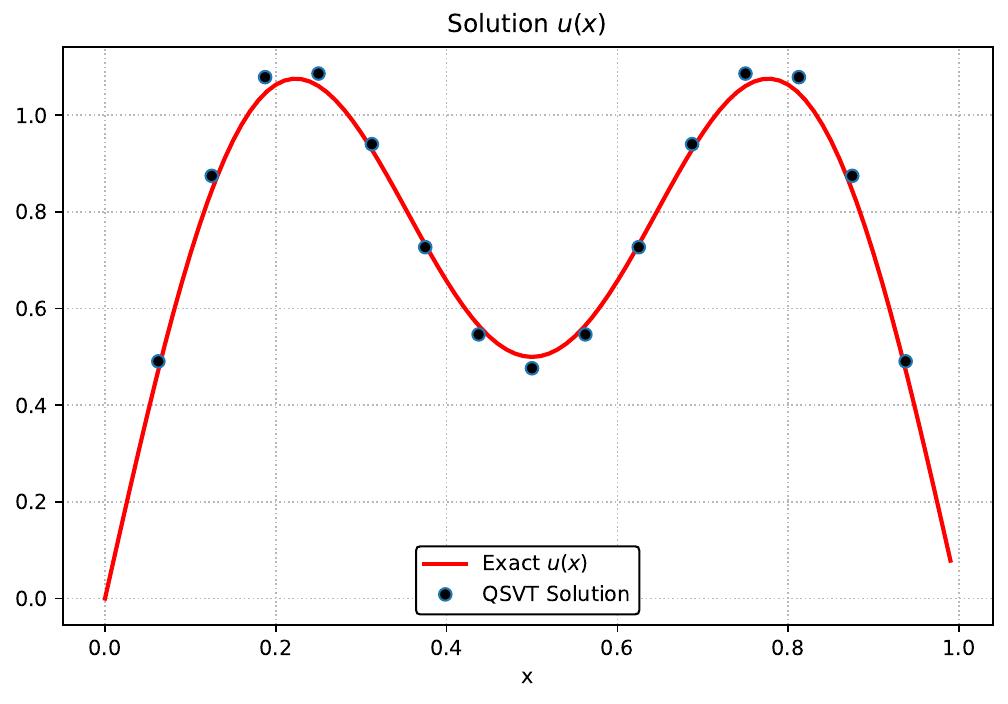}
    \caption{Simply supported problem \eqref{eq:num-ss}. The quantum simulation and the classical finite-difference solution have comparable errors at the same grid resolution.}
    \label{fig:dstbi}
\end{figure}

\subsection{Direct Dirichlet--Neumann discretization}\label{subsec:num-direct}

To check the convergence of the direct discretization, we consider the three-dimensional test problem with exact solution
\begin{equation}
    u(x,y,z)=x(1-x)\sin(\pi y)\sin(\pi z).
\end{equation}
The numerical results are summarized in \Cref{tab:3d-convergence}. Both the $L^2$ and $L^\infty$ errors exhibit second-order convergence.

\begin{table}[ht]
\centering
\caption{Convergence of the direct finite-difference discretization for a three-dimensional Dirichlet--Neumann test problem.}
\label{tab:3d-convergence}
\scriptsize
\setlength{\tabcolsep}{3pt}
\begin{tabular}{ccccc}
\toprule
$N$ & $h$ & $\norm{e}_{L^2}$ & order & $\norm{e}_{L^\infty}$ \\
\midrule
10 & $1.00\!\times\!10^{-1}$ & $3.44\!\times\!10^{-4}$ & -- & $1.24\!\times\!10^{-3}$ \\
15 & $6.67\!\times\!10^{-2}$ & $1.55\!\times\!10^{-4}$ & 1.9687 & $5.43\!\times\!10^{-4}$ \\
20 & $5.00\!\times\!10^{-2}$ & $8.74\!\times\!10^{-5}$ & 1.9853 & $3.14\!\times\!10^{-4}$ \\
25 & $4.00\!\times\!10^{-2}$ & $5.61\!\times\!10^{-5}$ & 1.9914 & $2.00\!\times\!10^{-4}$ \\
\bottomrule
\end{tabular}
\end{table}

\FloatBarrier
\subsection{Coupled Laplace formulation}\label{subsec:num-coupled}

This experiment complements the coupled formulation in \Cref{sec:coupled}. It checks the convergence of the finite-difference splitting discretization and reports the empirical scaling of the condition number \(\kappa(K)\) that enters the QSVT-VTAA cost discussion for \eqref{eq:split-linear}.

We consider a two-dimensional homogeneous Dirichlet--Neumann problem with exact solution
\begin{equation}
    u(x,y)=p(x)p(y),
    \qquad p(t)=t^2(1-t)^2.
\end{equation}
The source term is chosen so that \(\Delta^2u=f\).  The coupled Laplace system in \eqref{eq:split-linear} shows second-order convergence in the \(L^2\) norm, as shown in \Cref{fig:coupled-convergence}. The observed second order convergency verifies the validity of our approach. The measured condition number grows approximately like \(N^{3}\) in this two-dimensional test, as shown in \Cref{fig:coupled-kappa}. In this two-dimensional test, the coupled formulation exhibits a smaller empirical condition-number exponent than the direct fourth-order matrix. Since these plots quantify the numerical behavior of the coupled discretization and the associated condition number, they are reported together with the other validation experiments in this section.

\begin{center}
    \maybeincludegraphics[width=0.95\linewidth]{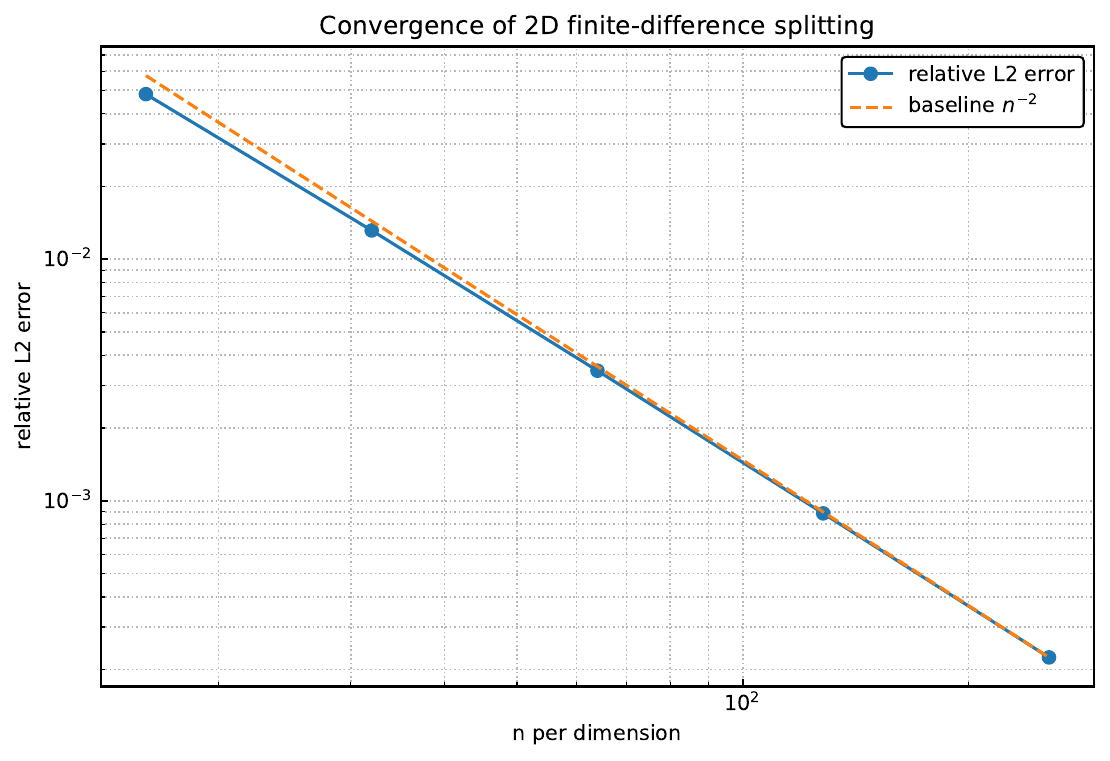}
    \captionof{figure}{Convergence order for the coupled Laplace finite-difference discretization.}
    \label{fig:coupled-convergence}
\end{center}

\begin{center}
    \maybeincludegraphics[width=0.95\linewidth]{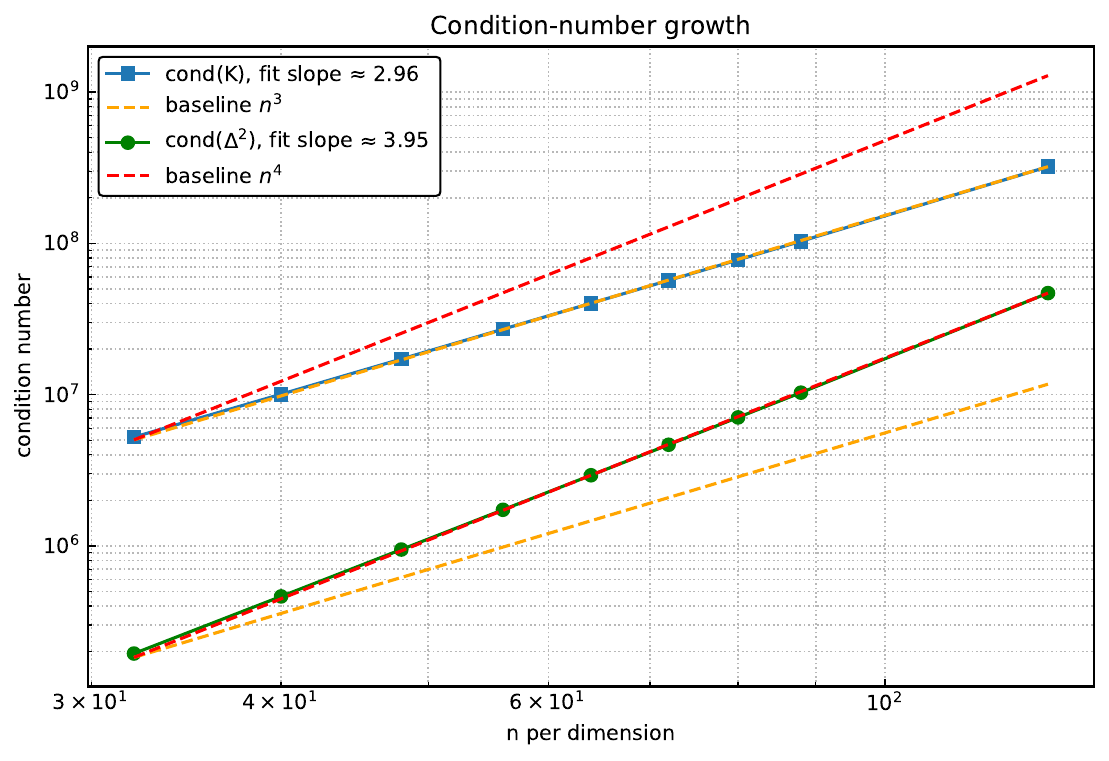}
    \captionof{figure}{Condition-number growth for the coupled Laplace system.}
    \label{fig:coupled-kappa}
\end{center}

\FloatBarrier
\section{Conclusion}\label{sec:conclusion}

We have developed QSVT--VTAA quantum linear-system formulations and
explicit block-encodings for several discretizations of the biharmonic
equation. For periodic and simply supported boundary conditions, the biharmonic equation can be reduced to two Poisson solves after diagonalization by the Fourier transform or the discrete sine transform. For Dirichlet--Neumann boundary conditions, we provided an explicit block-encoding of a direct fourth-order discretization and described an alternative coupled Laplace system. The analysis separates discretization error, block-encoding cost, condition-number dependence, and postselection cost. Numerical simulations confirm that the circuits reproduce the corresponding classical discretized solutions for small problem sizes.

Several questions remain for future work. First, a sharp stability theory for the splitting formulation under minimal assumptions would complete the condition-number analysis of the coupled system. Second, preconditioning may further reduce the condition number and complexity. Third, the block-encodings of boundary operators can likely be optimized further. Finally, implementing the proposed circuits on real quantum hardware or more detailed noise models would clarify their practical requirements.

\section*{Acknowledgements}
CM was supported by Fundamental and Interdisciplinary Disciplines Breakthrough Plan of the Ministry of Education of China (JYB2025XDXM112) and NSFC grant No.~12501607.  CM was also partially supported by the Science and Technology Commission of Shanghai Municipality (No.~22DZ2229014).

The authors thank {\it UnitaryLab}  for providing the
quantum-circuit design and classical simulation
platform; the circuit diagrams in this manuscript
were prepared using {\it UnitaryLab}.
\section*{Code availability}
The corresponding source code is publicly available in the GitHub repository at 
\url{https://github.com/TzhMo/Quantum-Biharmonic-Solver}

\end{document}